\documentclass[10pt,a4paper]{article}
\usepackage[margin=1in]{geometry}
\usepackage{graphicx}
\usepackage[colorlinks=true, allcolors=black]{hyperref}
\usepackage{amsmath,amsthm,amsfonts,amssymb,amscd,mathrsfs}
\usepackage{enumerate}
\usepackage{mathrsfs}
\usepackage{xcolor}
\usepackage{graphicx}
\usepackage{hyperref}
\usepackage{enumitem}
\usepackage{bbm}
\usepackage{float}
\usepackage{fancyvrb}
\usepackage{booktabs}
\usepackage{subcaption}
\usepackage{algpseudocode,algorithmicx,algorithm}
\usepackage{etoolbox}
\patchcmd{\abstract}{\small}{}{}{}
\usepackage{arydshln}

\def\algbackskip{\hskip-\ALG@thistlm}
\newtheorem{theorem}{Theorem}[section]
\newtheorem{definition}[theorem]{Definition}

\newtheorem{proposition}[theorem]{Proposition}
\newtheorem{lemma}[theorem]{Lemma}

\newtheorem{assumption}[theorem]{Assumption}
\newtheorem{problem}[theorem]{Problem}

\newtheorem*{theorem*}{Problem}
\newtheorem{remark}[theorem]{Remark}

\DeclareMathOperator{\rank}{rank}
\usepackage{mathrsfs}
\usepackage{mathtools}

\allowdisplaybreaks[4]
\title{Data-Enabled Policy and Value Iteration for Continuous-Time Linear Quadratic Output Feedback Control}
\date{\today}
\author{Jun Xie\footnote{College of Artificial Intelligence, Nankai University, China. Email: xiejun@mail.nankai.edu.cn}~~~~  Yuan-Hua Ni\footnote{College of Artificial Intelligence, Nankai University, China. Email: yhni@nankai.edu.cn}~~~~Yiqin Yang\footnote{The Key Laboratory of Cognition and Decision Intelligence for Complex Systems, Institute of Automation, Chinese Academy of Sciences, China. Email: yiqin.yang@ia.ac.cn}~~~~Bo Xu\footnote{The Key Laboratory of Cognition and Decision Intelligence for Complex Systems, Institute of Automation, Chinese Academy of Sciences, China. Email: xubo@ia.ac.cn}}

\begin{document}
	
	\maketitle
	
	\begin{abstract}

    This paper proposes efficient policy iteration and value iteration algorithms for the continuous-time linear quadratic regulator problem with unmeasurable states and unknown system dynamics, from the perspective of direct data-driven control.
    Specifically, by re-examining the data characteristics of input-output filtered vectors and introducing QR decomposition, an improved substitute state construction method is presented that further eliminates redundant information, ensures a full row rank data matrix, and enables a complete parameterized representation of the feedback controller.
    Furthermore, the original problem is transformed into an equivalent linear quadratic regulator problem defined on the substitute state with a known input matrix, verifying the stabilizability and detectability of the transformed system.
    Consequently, model-free policy iteration and value iteration algorithms are designed that fully exploit the full row rank substitute state data matrix. 
    The proposed algorithms offer distinct advantages: they avoid the need for prior knowledge of the system order or the calculation of signal derivatives and integrals; the iterative equations can be solved directly without relying on the traditional least-squares paradigm, guaranteeing feasibility in both single-output and multi-output settings; and they demonstrate superior numerical stability, reduced data demand, and higher computational efficiency. 
    Moreover, the heuristic results regarding trajectory generation for continuous-time systems are discussed, circumventing potential failure modes associated with existing approaches.

	\end{abstract}

	
	\section{Introduction}\label{sec1}	
    Optimal control is a cornerstone of modern control theory, and the Linear Quadratic Regulator (LQR) problem, where the system is linear and the cost functional is quadratic, occupies a central position. Featuring a canonical mathematical form, an explicit solution via the Riccati equation, and broad utility as an approximation for many nonlinear optimal control problems, LQR has been widely applied in robotics \cite{Robot}, autonomous driving \cite{Autonomous}, and other fields.  
    However, system states are often rendered unmeasurable by cost-effectiveness and engineering practicality \cite{Chen-LQG}, whereas system parameters are frequently unknown; these two factors substantially limit the direct applicability of conventional LQR approaches.
    Accordingly, this paper focuses on synthesizing the optimal feedback controller for the continuous-time LQR problem with unmeasurable states and unknown dynamics. 
    Given that the lack of model information poses difficulties for estimating unknown states using model-free methods \cite{Deng-CT-output},
    utilizing limited accessible information to solve this problem is quite challenging.

    In the discrete-time setting, the LQR problem with unmeasurable states and unknown parameter matrices has been extensively studied.
    Policy gradient methods are often used to solve for optimal static or dynamic output feedback controllers \cite{Sof,dLQR}; however, such methods face complex optimization structures, and the resulting controllers may not achieve the optimal performance of the LQR problem \cite{Deng-DT-output}.
    An alternative approach relies on finite-length history stacks \cite{Lewis-IOH} or input-output filtered vectors \cite{Deng-DT-output} to construct an injective representation of the unmeasurable states. Controllers derived from this approach yield performance equivalent to that of state feedback controllers \cite{State-Para}, which enables the development of direct data-driven policy iteration (PI) and value iteration (VI) algorithms. 
    However, the feasibility conditions of this approach are overly restrictive for multi-output problems, thereby limiting its scope of application.
    The work in \cite{DD-PIVI} leverages the powerful representational capability of data matrices under Willems' Fundamental Lemma \cite{Willems,De-Formulas} and develops a data-enabled framework extending \cite{Deng-DT-output}. This method eliminates the need for numerical solutions via least squares (LS), guaranteeing feasibility in multi-output scenarios.

   Since most physical systems inherently possess continuous-time dynamics, developing data-driven representations for continuous-time linear systems and solving the continuous-time LQR problem without relying on system models or state measurements is of greater significance. The traditional paradigm of discretizing the system before controller design often leads to suboptimal performance \cite{SAAA-CT-output}; accordingly, this paper aims to design optimal controllers directly for the continuous-time LQR problem.   
   If the states are directly measurable, the work in \cite{Jiang-2012} proposes model-free integral reinforcement learning to simultaneously estimate the cost parameter matrix and state feedback gain. By leveraging the continuous-time extension of Willems' Fundamental Lemma \cite{CT-PCPE}, the research in \cite{CT-DD-PI} converts the Hamilton-Jacobi-Bellman (HJB) equation into a Sylvester transpose equation, and proposes a more efficient integral PI algorithm. However, the absence of an input-output matrix construction method that serves as a substitute state matrix while preserving key properties prevents the direct extension of this approach to unmeasurable state scenarios.
   When states are unmeasurable, the work in \cite{Jiang-2016} proposes a model-free off-policy reinforcement learning algorithm with a discounted cost function to suppress exploration noise-induced bias. However, an inappropriately chosen discount factor may jeopardize system stability \cite{IRL-bias} and yield suboptimal LQR solutions; hence, methods based on discounted cost functions should be avoided. Drawing inspiration from discrete-time methods, the research in \cite{SAAA-CT-output} employs input-output filtered vectors as substitute states and derives a new output feedback HJB equation, upon which model-free integral PI and VI algorithms are established. Building on this, the work in \cite{Deng-CT-output} further introduces the replication of observation errors into the substitute states, enhancing the transient performance of the algorithm while guaranteeing closed-loop stability and solution optimality. Additionally, adopting the substitute states from \cite{SAAA-CT-output}, the work in \cite{Huang-newLQR} reformulates the original LQR problem into a steady-state equivalent one defined on the substitute states, and the iterative algorithms designed for the transformed LQR problem exhibit lower computational complexity.  
   It is important to note that the substitute states adopted in \cite{SAAA-CT-output,Huang-newLQR} are equivalent to the true states only at steady state; consequently, the corresponding feedback controller merely achieves steady-state performance equivalence. Therefore, it is necessary to use a stabilizing feedback controller with superimposed exploration noise to collect data over an extended period to obtain approximate steady-state data for algorithm implementation.
   Furthermore, the aforementioned algorithms \cite{SAAA-CT-output,Deng-CT-output,Huang-newLQR} all rely on the LS method for numerical solutions. To ensure solvability, the corresponding regression matrices are required to be full row rank, which is difficult to achieve for multi-output problems. Additionally, the introduction of Kronecker products significantly increases computational complexity and data requirements.

   To this end, this paper focuses on the continuous-time LQR problem with unmeasurable states and unknown dynamics. 
   To address the limitations of existing research, this paper investigates  three core issues, corresponding to Sections \ref{sec3}, \ref{sec4}, and \ref{sec5}, respectively:
   \begin{itemize}
    \item  Construct model-free effective substitute states such that the resulting data matrix is guaranteed to be full row rank and can be used to parametrically characterize the closed-loop systems and the controllers.
    
    \item Directly develop model-free, LS-free PI and VI algorithms that are feasible for both single-output and multi-output scenarios, while guaranteeing the optimality of the derived controller.  
    	
   	\item Provide a more universal model-free continuous-time trajectory generation result that corresponds to Willems' Fundamental Lemma, offering a new perspective for data-driven continuous-time system simulation. 	
   \end{itemize}
   	However, due to the instantaneous dynamics of continuous-time systems, the lack of a universal continuous-time counterpart to Willems' Fundamental Lemma, and the absence of an independent Q-function (whose role is subsumed by the V-function), these tasks are not straightforward extensions of discrete-time results \cite{DD-PIVI} to the continuous-time setting. 
	This paper revisits the substitute states proposed in \cite{Deng-CT-output} and reveals that, under the controllability assumption, the mapping matrix between substitute states and true states is guaranteed to be of full row rank. This not only ensures the unique reconstruction of true states from substitute states but also lays a critical foundation for establishing the equivalence between the proposed iterative algorithms and model-based state feedback iterative algorithms \cite{model-based-PI,Jiang-2016}.
	Furthermore, this paper analyzes the row rank property of the data matrix formed by the original substitute states under persistently exciting (PE) inputs, showing that it is inevitably rank-deficient in multi-output scenarios. To resolve this issue, we perform a QR decomposition on this data matrix and use the orthonormal basis matrix of the column space of the orthogonal matrix to project the original substitute states. The resulting improved substitute state data matrix is guaranteed to be of full row rank, a condition essential for ensuring the convergence of the data-enabled iterative algorithms.	
    This paper further proves that the improved substitute states satisfy a new linear system with the same input as the original system and a known input matrix. By reformulating the original quadratic cost functional into one involving the improved substitute states and inputs, an equivalent new LQR problem is constructed, whose optimal feedback controller achieves performance equivalent to that of the optimal state feedback controller for the original problem.
    Moreover, this paper demonstrates that this equivalent new LQR problem satisfies the stabilizability and detectability conditions required for iterative algorithms \cite{Huang-newLQR}. Accordingly, utilizing the full row rank improved substitute state data matrix, we develop directly solvable model-free HJB iterative equations (tractable generalized Sylvester equations) and model-free cost parameter matrix update equations, serving as the core modules for the data-enabled PI and VI algorithms, respectively.
    In addition, this paper shows that the full row rank improved substitute state data matrix can be used to construct an ordinary differential equation for the weighting vector in trajectory generation. This provides key insights for establishing the trajectory generation properties of the continuous-time counterpart to Willems' Fundamental Lemma.

	The main contributions of this paper are summarized as follows:
	\begin{itemize}
		\item 
		The proposed effective substitute state construction method that guarantees the resulting data matrix is of full row rank is applicable to various optimal control scenarios with unmeasurable states and provides an effective new tool for input-output-based data-driven control.
	    Furthermore, this paper transforms the LQR problem with unmeasurable states and unknown dynamics into an equivalent new LQR problem with known states and a known input matrix. This transformation not only ensures the optimality of the resulting controller but also reduces the design complexity of model-free algorithms. Compared with \cite{Huang-newLQR}, the proposed equivalent LQR problem further accounts for the influence of observation errors, yielding an optimal solution that is not confined to steady-state optimality.
		
	   \item This paper proposes data-enabled PI and VI algorithms for directly solving the continuous-time LQR problem with unmeasurable states and unknown dynamics. These algorithms eliminate the need for the LS method, relax feasibility conditions, and offer superior performance in terms of both data requirements and computational efficiency. Furthermore, by avoiding the numerical calculation of high-order derivatives or integrals of input-output signals during data collection, numerical errors are significantly reduced. Theoretically equivalent to model-based state feedback algorithms, the proposed algorithms guarantee stability, convergence, and optimality.
	   Compared with \cite{SAAA-CT-output,Deng-CT-output,Huang-newLQR}, the proposed algorithms do not require additional rank conditions on high-dimensional regression matrices, relying solely on the fundamental PE condition. Consequently, they are effective and feasible for both single-output and multi-output scenarios. In terms of data demand, the proposed algorithms only require the column number of the full row rank substitute state data matrix to be no less than its row number, reducing data requirements. Regarding computational efficiency, the projection operation removes unnecessary redundant information from the data matrix, and the proposed iterative equations can be solved directly, thereby lowering the computational burden.
	   In comparison to \cite{Deng-CT-output}, the proposed iterative equations only solve for the cost parameter matrix, not the feedback gain. Moreover, this work is not a direct extension of the state feedback data-enabled PI algorithm \cite{CT-PCPE}; benefiting from the transformation into an equivalent new LQR problem, the proposed PI iterative equation utilizes standard Sylvester equations (instead of the Sylvester transpose equations in \cite{CT-PCPE}), which exhibits superior stability properties.

	   \item This paper utilizes the proposed full row rank substitute state data matrix to provides constructive conditions for continuous-time trajectory generation, establishing a counterpart to the discrete-time Willems' Fundamental Lemma. The resulting trajectories are exact, in contrast to the approximate methods based on orthogonal basis expansions in \cite{ortho1,ortho2}. Additionally, for multi-output problems, the proposed method effectively avoids the unsolvability issue of the ordinary differential equations that restrict trajectory generation in \cite{In-out}, and it eliminates the need for numerical calculation of high-order derivatives.        
	\end{itemize}

	\textbf{Notation.}
	$\mathbb{R}$, $\mathbb{R}_{+}$, $\mathbb{N}$, $\mathbb{N}_{++}$, 	$\mathbb{S}^n$ and $\mathbb{S}^n_{++}$ ($\mathbb{S}^n_+$) denote the set of real numbers, the set of nonnegative real numbers, the set of natural numbers, the set of positive integers, the set of $n\times n$ symmetric matrices and the set of $n\times n$ positive (semi-)definite symmetric matrices, respectively. 
	$A^{\frac{1}{2}}$ means the unique symmetric positive semi-definite square root of $A\in\mathbb{S}_{+}^{n}$. 
	$I_n$ is defined as an $n\times n$ identity matrix, and $\mathbf{0}$ denotes a zero vector or matrix with appropriate dimension. 
	Let $\Vert\cdot\Vert$ denote the Euclidean vector norm or matrix spectral norm. For a matrix $A$, $A^\top$, $A^{\dagger}$, $\rank(A)$, and $\mathrm{im}(A)$ denote its transpose, Moore-Penrose pseudo-inverse, rank and image space, respectively. 
	For a square matrix $B$, $\mathrm{det}(B)$ and $\mathrm{adj}(B)$ denote its determinant and adjugate matrix, respectively. 
	Define $\mathrm{diag}(A,B)$ as the block diagonal matrix with main matrix blocks $A$ and $B$. 
	Let $s$, $\mathcal{L}(\cdot)$ and $\mathcal{L}^{-1}(\cdot)$ represent the Laplace variable, the Laplace transform and the inverse Laplace transform, respectively.
	For matrices $A$, $B$, and $C$ of appropriate dimensions, define
	$\mathcal{R}_N(A,B):=[A^{N-1}B,\cdots,AB,B]$, $\mathcal{O}_N(A,C):=[C^\top,(CA)^\top,\cdots,(CA^{N-1})^\top]^\top$,
	$$\mathcal{T}_N(A,B,C):=\begin{bmatrix}
		\mathbf{0}&&&\\
		CB&\mathbf{0}&&\\
		\vdots&\ddots&\ddots&\\
		CA^{N-2}B&\cdots&CB&\mathbf{0}
	\end{bmatrix},$$
	and the subscript $N$ will be omitted when the context is clear.
	For $A=[a_1,\cdots,a_m]\in\mathbb{R}^{n\times m}$, define $\mathrm{vec}(A):=[a_1^\top,\cdots, a_m^\top]^\top\in\mathbb{R}^{mn}$. The notation $\otimes$ is defined as the Kronecker product. 
	For a discrete-time signal sequence $\{u_t\}_{t\in\mathbb{N}}$, define its stacked window vector as $u_{[i,j]}:=[u_i^\top,u_{i+1}^\top,\cdots,u_{j}^\top]^\top$ with $i<j\in\mathbb{N}$, and a $N$-th order Hankel matrix is given by $\mathcal{H}_{N}(u_{[0,T+N-2]}):=[u_{[0,N-1]},u_{[1,N]},\cdots,u_{[T-1,T+N-2]}]$ with $N,T\in\mathbb{N}_{++}$.
	For a continuous-time signal $\{u_t\}_{t\in\mathbb{R}_{+}}$ that is sufficiently smooth, $u_t^{(i)}$ denotes the $i$-th order derivative of $u_t$; define the sampling matrix as $H(\{u_t\}_{t=t_1,\cdots,t_T}):=[u_{t_1},u_{t_2},\cdots,u_{t_T}]$ with $t_1<\cdots<t_{T}\in\mathbb{R}_{+}$, and the $N$-th order derivative sampling matrix is given by $H_N(\{u_t\}_{t=t_1,\cdots,t_T}):=[H(\{u_t\}_{t=t_1,\cdots,t_T})^\top,H(\{u_t^{(1)}\}_{t=t_1,\cdots,t_T})^\top,\cdots,H(\{u_t^{(N-1)}\}_{t=t_1,\cdots,t_T})^\top]^\top$; let $\Delta t>0$ represent the time delay step size, the delayed stacked window vector at time~$t$~is defined as $u_{(t,\Delta t,N)}:=[u_t^\top,\cdots,u_{t+(N-1)\Delta t}^\top]^\top$, and the $N$-th order delayed Hankel matrix at time~$t$~is given by $H_N(u_{(t,\Delta t, N+T-1)}):=[u_{(t,\Delta t,N)},u_{(t+\Delta t,\Delta t,N)},\cdots,u_{(t+(T-1)\Delta t,\Delta t,N)}]$ with $N, T\in\mathbb{N}_{++}$.
	
	\section{Problem Formulation and Preliminaries}\label{sec2}
	This section presents the problem setup and reviews existing research relevant to the subsequent analysis.
	\subsection{Problem Formulation}\label{sec2_1}
	Consider a continuous-time linear time-invariant system described by
	\begin{equation}\label{system}
		\begin{aligned}
			\dot{x}_{t}&=Ax_{t}+Bu_{t},\quad y_{t}=Cx_{t},
		\end{aligned}
	\end{equation}
	where $x_t\in\mathbb{R}^{n}$ is the unmeasurable state, $u_t\in\mathbb{R}^{m}$ and $y_t\in\mathbb{R}^{p}$ are the measurable input and output, respectively, with $p\leq n$. Without loss of generality, the output matrix $C$ is assumed to be of full row rank. Let $\mathcal{U}_{ad}$ denote the set of admissible controllers.
	
	Define the infinite-horizon quadratic cost functional as
	\begin{equation}\label{J}
		\begin{aligned}
			J(x_{0},u)&=\int_{0}^{\infty}\left(y_{t}^\top Qy_{t}+u_{t}^\top Ru_{t}\right)dt,
		\end{aligned}
	\end{equation}
	where $Q\in\mathbb{S}^{p}_{+}$ and $R\in\mathbb{S}^{m}_{++}$ are constant weighting matrices. Furthermore, define the state weighting matrix as $Q_x:=C^\top QC\in\mathbb{S}^{n}_{+}$.
	
	This paper aims to design the optimal feedback controller for the continuous-time LQR problem with unmeasurable states and unknown model parameters. The formal problem statement is given as follows.
	\begin{problem}\label{P1}
		Considering system (\ref{system}), where $A$, $B$ and $C$ are unknown, and states $\{x_t\}$ are unmeasurable, find an optimal control policy $u^*$ such that
		\begin{equation}\label{u*}
			u^*=\arg\min_{u\in\mathcal{U}_{ad}} J(x_{0},u).
		\end{equation}
	\end{problem}
	This paper adopts the following standard assumptions for solving the unmeasurable-state LQR problem using iterative algorithms.
	\begin{assumption}[Controllability and Observability.]\label{as1}
		 The pair $(A, B)$ is controllable, and the pairs $(A, C)$, $(A, Q_x^{\frac{1}{2}})$ are observable.
	\end{assumption}
	
	\begin{remark}
		In general, the stabilizability of $(A,B)$ is sufficient to guarantee the existence of the optimal LQR controller; combined with the observability of $(A,C)$, it also ensures that the states of system (\ref{system}) can be uniquely reconstructed from input-output data \cite{Deng-CT-output}. However, this paper aims to utilize data matrices to develop efficient model-free algorithms. As discussed later, the desirable properties of data matrices required for theoretical analysis are more readily derived under the controllability assumption. Therefore, consistent with classical literature \cite{De-Formulas,Efficient-Q}, we further assume that $(A,B)$ is controllable. Additionally, the observability of $(A,Q_x^{\frac{1}{2}})$ is a necessary prerequisite for the convergence of the iterative algorithms \cite{Jiang-2012,Jiang-iter}.
	\end{remark}
	
	\subsection{Model-Based State Feedback LQR}\label{sec2_2}
	For a known system with measurable states, the optimal LQR controller is $u_t=-K_x^*x_t$, where
	\begin{equation}\label{K*}
		K_x^*=R^{-1}B^\top P_x^*,
	\end{equation}
	and $P_x^*\in\mathbb{S}^n_{++}$ is the unique solution to the following algebraic Riccati equation (ARE)
	\begin{equation}\label{ARE}
		A^\top P^*_x+P^*_xA+Q_x-P_x^*BR^{-1}B^\top P^*_x=\mathbf{0}.
	\end{equation}
	If $K_x$ is a stabilizing state feedback gain such that $A-BK_x$ is Hurwitz stable, then $J(x_0,-K_xx)=x_0^\top P_x x_0$, where $P_x\in\mathbb{S}^n_{++}$ is the unique solution to the following Lyapunov equation
	\begin{equation}\label{lyap}
		(A-BK_x)^\top P_x+P_x(A-BK_x)+Q_x+K_x^\top RK_x=\mathbf{0}.
	\end{equation}
	If an initial stabilizing state feedback gain $K_x^0$ is selected and iterations are performed on \eqref{lyap} and \eqref{K*} until convergence, the classical Kleinman's algorithm is obtained, which possesses desirable stability and convergence properties \cite{model-based-PI}. Conversely, if the ARE \eqref{ARE} is recursively updated with a suitable step size, the classical model-based VI algorithm is obtained \cite{Jiang-2016}.

	 \subsection{Willems' Fundamental Lemma for Discrete-Time Systems}\label{sec2_3}
	\begin{definition}[Discrete-Time PE \cite{Efficient-Q}]\label{DT-PE}
		A discrete-time signal sequence $\{u_t\}_{t\in\mathbb{N}}$ is PE of order $N$ if the Hankel matrix $\mathcal{H}_{N}(u_{[0,N+T-2]})$ has full row rank, where $T\geq mN$.
	\end{definition}

	Leveraging the favorable structures of the Hankel matrix and discrete-time linear time-invariant systems, the following important conclusions regarding the properties of data matrices and system trajectory generation can be derived.
	
	\begin{lemma}[Willems' Fundamental Lemma \cite{Willems}]\label{willems}
		Consider a controllable discrete-time linear time-invariant system $x_{t+1}=Ax_t+Bu_t$, $y_t=Cx_t$. If a PE input signal $\{u_t\}_{t\in\mathbb{N}}$ of order $n+N$ is applied, generating the corresponding state and output sequences $\{x_t\}_{t\in\mathbb{N}}$ and $\{y_t\}_{t\in\mathbb{N}}$, then the following assertions hold.
		\begin{itemize}
			\item[\textup{(a)}] 	$\rank\left(\begin{bmatrix}
				\mathcal{H}_{N}(u_{[0,T+N-2]})\\\mathcal{H}_{1}(x_{[0,T-1]})
			\end{bmatrix}\right)=mN+n$.
			\item[\textup{(b)}] 
			Let $N\geq \ell_{ob}$, where $\ell_{ob}$ is the observability index of the linear system. The sequence $[\bar{u}_{[0,N-1]}^\top,\bar{y}_{[0,N-1]}^\top]^\top$ is an input-output trajectory of length $N$ for the discrete-time linear system, if and only if there exists a vector $\alpha\in\mathbb{R}^{T}$ such that
			\begin{equation}\label{alp}
				\begin{bmatrix}
					\mathcal{H}_N(u_{[0,T+N-2]})\\\mathcal{H}_N(y_{[0,T+N-2]})
				\end{bmatrix}\alpha=\begin{bmatrix}
					\bar{u}_{[0,N-1]}\\\bar{y}_{[0,N-1]}
				\end{bmatrix}.
			\end{equation}
		\end{itemize}
	\end{lemma}
	A corollary of assertion (a) is given in \cite{De-Formulas}:
	\begin{equation}\label{IK}
		\mathrm{im}\left(\begin{bmatrix}
			K_x\\I_n
		\end{bmatrix}\right)\in\mathrm{im}\left(\begin{bmatrix}
			\mathcal{H}_1(u_{[0,T-1]})\\\mathcal{H}_1(x_{[0,T-1]})
		\end{bmatrix}\right).
	\end{equation} 
	Thus, assertion (a) is often used to parameterize closed-loop systems and state feedback control gains in model-free optimal control problems; assertion (b) is commonly employed for model-free trajectory generation in model predictive control \cite{DeePc}.

	\subsection{Data Characteristics of Continuous-Time Systems}\label{sec2_4}
	Owing to the structural similarity between the differential operator of continuous-time linear systems and the shift operator of discrete-time linear systems, the full row rank property of continuous-time signal matrices composed of inputs, their high-order derivatives, and states can be obtained analogously \cite{CT-PE}. 
	However, due to the essential differences between these operators, the corresponding results for system trajectory generation in continuous time remain elusive. 
	\begin{definition}[Continuous-Time PE \cite{CT-PE}]\label{CT-PE}
		 Let $\mathbb{I}\subseteq\mathbb{R}_{+}$ be an open time interval. A continuous-time signal $\{u_t\}_{t\in\mathbb{I}}$ is PE of order $N$ if $u_t$ is $N-1$ times continuously differentiable on $\mathbb{I}$, and $\nu^\top[u_t^\top,u_t^{(1)\top},\cdots,u_t^{(N-1)\top}]^\top=\mathbf{0}$ holds for all $t\in\mathbb{I}$ if and only if $\nu=\mathbf{0}$.
	\end{definition}
	
	\begin{lemma}[\cite{CT-PE}]\label{dot-ux}
		Consider the controllable continuous-time linear time-invariant system (\ref{system}). If the input $\{u_t\}_{t\in\mathbb{I}}$ is PE of order $n+N$ and $\{x_t\}_{t\in\mathbb{I}}$ is the corresponding state signal, then $\nu^\top[u_t^\top,u_t^{(1)\top},\cdots,u_t^{(N-1)\top},x_t^\top]^\top=\mathbf{0}$ holds for all $t\in\mathbb{I}$ if and only if $\nu=\mathbf{0}$.
	\end{lemma}
	
	The work in \cite{CT-PE} further shows that if the signal is sampled at arbitrary fixed time instants $t_1<t_2<\cdots<t_T\in\mathbb{I}$ with $T\geq mN+n$, then under PE inputs of order $n+N$, the input-derivative-state data matrix has full row rank, i.e.,
	\begin{equation}\label{Hux}
		\mathrm{rank}\left(\begin{bmatrix}
			H_{N}(\{u_t\}_{t=t_1,\cdots,t_T})\\H_{1}(\{x_t\}_{t=t_1,\cdots,t_T})
		\end{bmatrix}\right)=mN+n.
	\end{equation}
    For the data matrix composed of inputs, outputs, and their high-order derivatives, under PE inputs of order $n+N$, it follows from 
	\begin{equation}\label{uy-ux}
		\begin{bmatrix}
			H_{N}(\{u_t\}_{t=t_1,\cdots,t_T})\\H_{N}(\{y_t\}_{t=t_1,\cdots,t_T})
		\end{bmatrix}=\begin{bmatrix}
			I_{mN}&\mathbf{0}\\\mathcal{T}(A,B,C)&\mathcal{O}(A,C)
		\end{bmatrix}\begin{bmatrix}
			H_{N}(\{u_t\}_{t=t_1,\cdots,t_T})\\H_{1}(\{x_t\}_{t=t_1,\cdots,t_T})
		\end{bmatrix},
	\end{equation}
	and the observability of system (\ref{system}) that
	\begin{equation}\label{rank-uy}
		\mathrm{rank}\left(\begin{bmatrix}
			H_{N}(\{u_t\}_{t=t_1,\cdots,t_T})\\H_{N}(\{y_t\}_{t=t_1,\cdots,t_T})
		\end{bmatrix}\right)=mN+n.
	\end{equation}
	For simplicity, uniform sampling with sampling interval $\Delta t$ is adopted hereinafter.

\section{An Effective Substitution State and Equivalent LQR Problem}\label{sec3}
This section addresses two main problems: first, how to design an effective substitute state in a model-free manner, such that it maintains an injective linear relationship with the true unmeasurable state and the data matrix formed by its samples is of full row rank; second, how to utilize this substitute state to convert the information-limited original LQR Problem \ref{P1} into an equivalent new LQR problem with fewer unknowns.

\subsection{Data Characteristics of Effictive Substitute State}\label{sec3_1}
When the states $\{x_t\}$ of system (\ref{system}) are unmeasurable, an $n$-th order filter of the input-output signals can be used as a substitute state based on a Luenberger observer \cite{Deng-CT-output}. However, if the goal is to parameterize feedback gains using the data matrix formed by such a substitute state under PE inputs, the row rank property of this data matrix must be re-examined. This subsection addresses this issue and proposes an effective substitute state construction method that strictly guarantees the resulting data matrix is of full row rank.

Specifically, the estimated state $\hat{x}_t$ obtained via a Luenberger observer satisfies the system
\begin{equation}\label{hatx}
	\dot{\hat{x}}_t=(A-LC)\hat{x}_t+Bu_t+Ly_t,
\end{equation}
where $L$ is the Luenberger observer gain, and $\hat{x}_0=\mathbf{0}$ is typically chosen. The observation error $\epsilon_t:=x_t-\hat{x}_t$ satisfies the system
\begin{equation}\label{eps}
	\dot{\epsilon}_t=(A-LC)\epsilon_t.
\end{equation}
Since $\epsilon_t=e^{(A-LC)t}\epsilon_0$ with $\epsilon_0\neq\mathbf{0}$ and $A-LC$ is Hurwitz stable, $x_t=\hat{x}_t$ holds only as $t\rightarrow \infty$. Furthermore, since the system matrices $A-LC$, $B$ and $L$ in (\ref{hatx}) are unknown, the estimated state $\hat{x}_t$ cannot be obtained in a model-free manner. Therefore, the work in \cite{Deng-CT-output} jointly considers systems (\ref{hatx}) and (\ref{eps}), obtaining model-free substitute vectors for $\hat{x}_t$ and $\epsilon_t$ via the Laplace operator. This yields a substitute vector $\eta_t$ for $x_t=\hat{x}_t+\epsilon_t$, as defined in Theorem \ref{thm1}, which can be interpreted as an $n$-th order filter of the input-output signals. Theorem \ref{thm1} further establishes that, under the controllability assumption, the linear mapping matrix $S$ from $\eta_{t}$ to $x_t$ (referred to as the state parameterization matrix) is of full row rank. To present a concise formal statement of these results in Theorem \ref{thm1}, the following notations are first defined:
$$\Lambda(s):=\mathrm{det}(sI_n-(A-LC)):=s^n+a_{n-1}s^{n-1}+\cdots+a_1s+a_0,$$
$$\mathrm{adj}(sI_n-(A-LC)):=D_{n-1}s^{n-1}+\cdots+D_{1}s^{1}+D_0,$$
and let $a_n=1$. From the relationship between the determinant and the adjugate matrix, it follows that $D_i=\sum_{j=1}^{n-i}a_{n-j+1}(A-LC)^{n-i-j}$ for $i=0,\cdots,n-1$, and define $\mathcal{D}:=[D_0,D_1,\cdots,D_{n-1}]$. For $i=1,\cdots,m$ and $j=1,\cdots,p$, let $B_i$ and $L_j$ denote the $i$-th column of $B$ and the $j$-th column of $L$, respectively, and further define $S^u_i=\mathcal{D}(I_n\otimes B_i)\in\mathbb{R}^{n\times n}$, $S^y_j=\mathcal{D}(I_n\otimes L_j)\in\mathbb{R}^{n\times n}$; let $U_i(s)$ and $Y_j(s)$ represent the Laplace transforms of the $i$-th component of $u_t$ and the $j$-th component of $y_t$, respectively, and further define
$$\eta_t^{u^i}:=\mathcal{L}^{-1}\left(\left[\frac{U_i(s)}{\Lambda(s)},\frac{sU_i(s)}{\Lambda(s)},\cdots,\frac{s^{n-1}U_i(s)}{\Lambda(s)}\right]^{\top}\right)\in\mathbb{R}^{n},$$
$$\eta_t^{y^j}:=\mathcal{L}^{-1}\left(\left[\frac{Y_j(s)}{\Lambda(s)},\frac{sY_j(s)}{\Lambda(s)},\cdots,\frac{s^{n-1}Y_j(s)}{\Lambda(s)}\right]^{\top}\right)\in\mathbb{R}^{n}.$$
Define $S^{u}:=[S_1^{u},\cdots,S_m^{u}]$, $S^{y}:=[S_1^{y},\cdots,S_p^{y}]$, $\eta_t^{u}:=[(\eta_t^{u^1})^\top,\cdots,(\eta_t^{u^m})^\top]^\top$, $\eta_t^{y}:=[(\eta_t^{y^1})^\top,\cdots,(\eta_t^{y^p})^\top]^\top$. Let $A_{\epsilon}$ be an arbitrary square matrix with the same eigenvalues as $A-LC$, and define $\mathcal{D}_{\epsilon}$ analogously to $\mathcal{D}$ by replacing $A-LC$ with $A_{\epsilon}$. Define
$$A_s:=\left[\begin{array}{c:c}
	\mathbf{0}&I_{n-1}\\\hdashline
	-a_0&\begin{matrix}
		-a_1&\cdots&-a_{n-1}
	\end{matrix}
\end{array}\right],\quad b_s:=\left[\begin{array}{c}
	\mathbf{0}_{(n-1)\times1}\\\hdashline1
\end{array}\right], $$
$\mathcal{A}_s:=\mathrm{diag}(I_{m+p}\otimes A_s,A_{\epsilon})$, $\mathcal{B}_{s}:=\mathrm{diag}(I_{m+p}\otimes b_s,\mathbf{0}_{n\times n})$.
	
\begin{theorem}\label{thm1}
	Consider system (\ref{system}). The following statements hold.
	\begin{itemize}
		\item [\textup{(a)}] The state $x_t$ of system (\ref{system}) can be represented by the input-output $n$-th order filter vector $\eta_t\in\mathbb{R}^{(m+p+1)n}$, i.e.,
		\begin{equation}\label{eta}
			x_t=[S^u,S^y,S^{\epsilon}][(\eta_t^u)^\top,(\eta_t^y)^\top,(\eta_t^{\epsilon})^\top]^\top:=S\eta_t,
		\end{equation}
		where $\eta_t^{\epsilon}\in\mathbb{R}^{n}$ is defined by the following system,
		\begin{equation}\label{eta-eps}
			\dot{\eta}_t^{\epsilon}=A_{\epsilon}\eta_t^{\epsilon},\quad \eta_0^{\epsilon}\neq\mathbf{0},
		\end{equation}
		and $S^{\epsilon}:=S^{\epsilon_x}(S^{\epsilon_{\eta}})^{-1}$, $S^{\epsilon_x}:=\mathcal{D}(I_n\otimes \epsilon_0)$, $S^{\epsilon_{\eta}}:=\mathcal{D_{\epsilon}}(I_n\otimes \eta_0^{\epsilon})$.
		The substitute state $\eta_t$ can be generated in a model-free manner via the following system.
		\begin{equation}\label{eta-sys}
			\dot{\eta}_t=\mathcal{A}_s\eta_t+\mathcal{B}_s[u_t^\top,y_t^\top, \mathbf{0}_{1\times n}]^\top,\quad \eta_0=[\mathbf{0}_{1\times (m+p)n},(\eta_0^{\epsilon})^\top]^\top.
		\end{equation}
		\item[\textup{(b)}] Under the controllability assumption on $(A,B)$, the state parameterization matrix $S$ has full row rank.
	\end{itemize}
\end{theorem}	
	
\begin{remark}
	A detailed proof of assertion (a) in Theorem \ref{thm1} is provided in \cite{Deng-CT-output}. By appropriately selecting the structure of $A_{\epsilon}$, $S^{\epsilon_{\eta}}$ can be made invertible. Note that generating $\eta_{t}$ via equation (\ref{eta-sys}) does not require knowledge of $A-LC$; instead, by simply assigning stable eigenvalues to $A-LC$, the block diagonal matrices $\mathcal{A}_s$ and $\mathcal{B}_s$ can be computed without any knowledge of system (\ref{system}).
\end{remark}

\begin{remark}	
	By the formal correspondence between the continuous-time Laplace operator and the discrete-time Z-transform operator, the continuous-time state parameterization matrix $S$ takes the same form as its discrete-time counterpart. Furthermore, since the controllability criteria are identical, the proof technique developed for discrete-time systems in \cite{DD-PIVI} can be directly applied to assertion (b) of Theorem \ref{thm1}.
	Reference \cite{Huang-newLQR} presents similar results for continuous-time systems, but does not involve the replication of the observation error system (\ref{eta-sys}). If $(A,B)$ is uncontrollable, \cite{State-Para,Deng-CT-output} show that $S$ can still be guaranteed to be of full row rank by either assigning eigenvalues of $A-LC$ disjoint from those of $A$ to ensure the controllability of $(A-LC,L)$, or by assigning distinct eigenvalues for $A-LC$ to ensure the controllability of $(A-LC,\epsilon_0)$.
\end{remark}

Note that the input to system (\ref{eta-sys}) comprises the input and output of the original system (\ref{system}),  together with zero vectors. Consequently, even if the input to system (\ref{system}) satisfies the PE condition, the input to system (\ref{eta-sys}) generally does not, due to the decaying excitation of the output and the presence of zero vectors. This prevents the direct application of the results in Section \ref{sec2_4} to analyze the row rank property of the data matrix $E_0:=[\eta_t,\cdots,\eta_{t+(T-1)\Delta t}]$, which is critical for characterizing the closed-loop systems and controllers, as well as for designing model-free iterative equations independent of LS. To resolve this issue, Lemma \ref{le1} employs $y_t=CS\eta_t$ to transform system (\ref{eta-sys}) into an isosystem of (\ref{system}) --- sharing the same input-output behavior but with state $\eta_t$. Then by analyzing the space spanned by $\eta_t$ within this isosystem, the row rank property of $E_0$ can be established.
Herein, the space spanned by $\eta_{t}$ is defined as the space formed by all possible $\eta_t$ generated under PE inputs, i.e., 
$\mathcal{V}_t := \left\{ \eta_t \mid \eta_t \text{ solves system } \eqref{eta-sys} \text{ for some } 2n\text{-th order PE inputs } u_t \text{ over the interval } [0, t] \right\}$. In the subsequent analysis, the space spanned by the state vector of any non-autonomous system is defined analogously, while that for an autonomous system does not require the PE condition on
$u_t$.

	\begin{lemma}\label{le1}
	Consider the substitute state $\eta_{t}$ generated by (\ref{eta-sys}). The following assertions hold.
	\begin{itemize}
		\item[\textup{(a)}] $\eta_t$ satisfies the isosystem given by
		\begin{equation}\label{C-sys}
			\begin{split}
				\dot{\eta}_{t}=(\mathcal{A}_s+\mathcal{B}_s^yCS)\eta_t+\mathcal{B}_s^uu_t, \quad y_t=CS\eta_t,
			\end{split}
		\end{equation}
		where $\mathcal{B}_s^u$ and $\mathcal{B}_s^y$ denote the matrices formed by the first $m$ columns and columns $m+1$ to $m+p$ of $\mathcal{B}_s$, respectively.
		\item[\textup{(b)}] Under Assumption \ref{as1}, if the input to system (\ref{system}) satisfies the $2n$-th order PE condition, the dimension of the space spanned by $\eta_t$ is $(m+2)n$. Moreover, the rank of the data matrix $E_0:=[\eta_t,\cdots,\eta_{t+(T-1)\Delta t}]$ is $(m+2)n$, and $E_0$ has full row rank if and only if $p=1$.
	\end{itemize}
\end{lemma}
	
\begin{proof}
	Assertion (a) follows immediately by substituting $y_t=CS\eta_t$ into the input vector of system (\ref{eta-sys}). We now prove assertion (b). Note that while the substitute state $\eta_{t}$ in (\ref{eta-sys}) is generated by filtering each component of the input and output individually, we herein consider filtering the entire vectors of the input and output for theoretical analysis. That is, define
	$$g_{t}:=F^{\eta}\eta_{t}=\left[\mathcal{L}^{-1}\left(\frac{U(s)}{\Lambda(s)}^\top,\frac{sU(s)}{\Lambda(s)}^\top,\cdots,\frac{s^{n-1}U(s)}{\Lambda(s)}^\top,\frac{Y(s)}{\Lambda(s)}^\top,\frac{sY(s)}{\Lambda(s)}^\top,\cdots,\frac{s^{n-1}Y(s)}{\Lambda(s)}^\top\right),\eta_t^{\epsilon\top}\right]^{\top},$$
    where $F^{\eta}$ is a row permutation matrix. The dimensions of the spaces spanned by $\eta_{t}$ and $g_{t}$ are therefore equal. To analyze the dimension of the space spanned by $\eta_t$, it suffices to consider that of $g_t$. Similarly, $g_t$ is partitioned into input, output, and observation error components as $g_t=[(g_t^u)^\top,(g_t^y)^\top,(\eta_t^{\epsilon})^\top]^\top$. Let $G^u(s):=\mathcal{L}(g_{t}^u)$; it follows that
	\begin{equation}\label{G}
		\Lambda(s)G^u(s)=[U(s)^\top,(sU(s))^\top,\cdots,(s^{n-1}U(s))^\top]^\top.
	\end{equation}
	Taking the inverse Laplace transform of both sides of equation (\ref{G}) yields
	$$a_0g_t^u+a_1(g_t^{u})^{(1)}+\cdots+a_{n-1}(g_t^{u})^{(n-1)}+(g_t^{u})^{(n)}=[(u_t)^\top,(u_t^{(1)})^\top,\cdots,(u_t^{(n-1)})^\top]^\top,$$
	Since $g_{t}^{u}$ is a block vector, let $g_{t}^{u}=[(\xi_0)^\top,\cdots,(\xi_{n-1})^\top]^\top$ with $\xi_j=\mathcal{L}^{-1}\left(\frac{s^jU(s)}{\Lambda(s)}\right)\in\mathbb{R}^{m}$ for $j=0,\cdots,n-1$. It follows that
	\begin{equation}\label{xi-u}
		a_0\xi^{(0)}_j+a_1\xi^{(1)}_j+\cdots+a_{n-1}\xi^{(n-1)}_j+\xi^{(n)}_j=u_t^{(j)},\quad j=0,\cdots,n-1.
	\end{equation}
	By the continuous differentiability of the input signal, it holds that $\xi_{j+1}=\xi_j^{(1)}$ for $j=0,1,\cdots,n-2$. Let $z_t:=\xi_0$; then $g_t^u=[(z_t)^\top,(z_t^{(1)})^\top,\cdots,(z_t^{(n-1)})^\top]^\top$. Since the mapping $u_t\mapsto z_t=\mathcal{L}^{-1}\left(\frac{U(s)}{\Lambda(s)}\right)$ is linear and invertible on the space of sufficiently smooth trajectories generated by PE signals, the dimensions of the spaces spanned by $g_t^{u}$ and $[(u_t)^\top,(u_t^{(1)})^\top,\cdots,(u_t^{(n-1)})^\top]^\top$ are equal. Similarly, the dimensions of the spaces spanned by $g_t^y$ and $[(y_t)^\top,(y_t^{(1)})^\top,\cdots,(y_t^{(n-1)})^\top]^\top$ are equal.
	In conclusion, to analyze the dimension of the space spanned by $g_t$, it suffices to consider that of $[(u_t)^\top,\cdots,(u_t^{(n-1)})^\top,(y_t)^\top,\cdots,(y_t^{(n-1)})^\top,(\eta^\epsilon_t)^\top]^\top$. Note that the following identity holds
	$$\begin{bmatrix}
		\begin{array}{c}
			u_t\\\vdots\\u_t^{(n-1)}\\\hdashline y_t\\\vdots\\y_t^{(n-1)}\\\hdashline\eta_t^{\epsilon}
		\end{array}
	\end{bmatrix}=\underbrace{\begin{bmatrix}
			\begin{array}{c:c:c}
				I_{mn}&\mathbf{0}&\mathbf{0}\\\hdashline
				\mathcal{T}(A,B,C)&\mathcal{O}(A,C)&\mathcal{O}(A,C)S^{\epsilon}\\\hdashline
				\mathbf{0}&\mathbf{0}&I_n
			\end{array}
	\end{bmatrix}}_{\mathcal{M}}\begin{bmatrix}
		\begin{array}{c}
			u_t\\\vdots\\u_t^{(n-1)}\\\hdashline \hat{x}_{t}\\\hdashline\eta_t^{\epsilon}
		\end{array}
	\end{bmatrix}.$$
    Under the controllability assumption for system (\ref{system}) and by Lemma \ref{dot-ux}, the vector $[(u_t)^\top,\cdots,(u_t^{(n-1)})^\top,x_t^\top]$ spans $\mathbb{R}^{mn+n}$ for inputs satisfying the $2n$-th order PE condition. As for the space spanned by $\eta^\epsilon_t$, since we have chosen $\mathcal{D}_{\epsilon}(I_n\otimes \eta_0^{\epsilon})$ to be invertible, the structure of $\mathcal{D}_{\epsilon}$ implies that the matrix $[A_{\epsilon}^{n-1}\eta_0^{\epsilon},\cdots,A_{\epsilon}\eta_0^{\epsilon},\eta_0^{\epsilon}]$ is invertible. Consequently, under the autonomous system (\ref{eta-eps}), the space spanned by $\eta_t^\epsilon$ is $\mathrm{im}([A_{\epsilon}^{n-1}\eta_0^{\epsilon},\cdots,A_{\epsilon}\eta_0^{\epsilon},\eta_0^{\epsilon}])=\mathbb{R}^{n}$.
	Since $\eta_t^{\epsilon}$ is independent of the inputs $\{u_t\}$ and the initial state $x_0$, $\eta_{t}^{\epsilon}$ is independent of $[(u_t)^\top,\cdots,(u_t^{(n-1)})^\top,(x_t)^\top]^\top$. Accordingly, the space spanned by $[(u_t)^\top,\cdots,(u_t^{(n-1)})^\top,(x_t)^\top,(\eta_t^\epsilon)^\top]^\top$ is $\mathbb{R}^{(m+2)n}$. Since $\hat{x}_t=x_t-S^{\epsilon}\eta_t^{\epsilon}$, it holds that
	$$\begin{bmatrix}
		\begin{array}{c}
			u_t\\\vdots\\u_t^{(n-1)}\\\hdashline \hat{x}_t\\\hdashline\eta_t^{\epsilon}
		\end{array}
	\end{bmatrix}=\underbrace{\begin{bmatrix}
			\begin{array}{c:c:c}
				I_{mn}&\mathbf{0}&\mathbf{0}\\\hdashline
				\mathbf{0}&I_n&-S^{\epsilon}\\\hdashline
				\mathbf{0}&\mathbf{0}&I_n
			\end{array}
	\end{bmatrix}}_{\mathcal{I}}\begin{bmatrix}
		\begin{array}{c}
			u_t\\\vdots\\u_t^{(n-1)}\\\hdashline x_{t}\\\hdashline\eta_t^{\epsilon}
		\end{array}
	\end{bmatrix},$$
	where $\mathcal{I}$ is invertible. Thus, the span of $[(u_t)^\top,\cdots,(u_t^{(n-1)})^\top,(\hat{x}_t)^\top,(\eta_t^\epsilon)^\top]^\top$ is $\mathbb{R}^{(m+2)n}$, meaning this vector is reachable. Consequently, the span of $[(u_t)^\top,\cdots,(u_t^{(n-1)})^\top,(y_t)^\top,\cdots,(y_t^{(n-1)})^\top,(\eta^\epsilon_t)^\top]^\top$  is $\mathrm{im}(\mathcal{M})$. Furthermore, under the observability assumption for system (\ref{system}), $\mathcal{M}$ has full column rank with rank $(m+2)n$. Combining these results, the dimension of the span of $\eta_t$ is $(m+2)n$. Therefore, under sufficiently exciting input signals, the rank of the substitute state data matrix $E_0$ is $(m+2)n$. However, since the number of its rows is $(m+p+1)n$, these two numbers are equal if and only if $p=1$.	
\end{proof}

\begin{remark}
	The substitute state $\eta_{t}$ is still obtained in a model-free manner via equation (\ref{eta-sys}). The isosystem (\ref{C-sys}) is introduced primarily to exploit its structural consistency with the original system (\ref{system}), facilitating algorithm design. Note that the dynamic matrix $\mathcal{A}_s+\mathcal{B}_s^yCS$ of the isosystem is unknown, whereas the input matrix $\mathcal{B}_s^u$ is known.
\end{remark}

\begin{remark}[Relationships and Comparisons with \cite{In-out}]\label{re1}
	From the proof of Lemma \ref{le1}, equation (\ref{xi-u}) for $j=0$ becomes
	$$u_t=a_0z_t+a_1z_t^{(1)}+\cdots+a_{n-1}z_t^{n-1}+z_t^{(n)},$$
	where $z_t=\mathcal{L}^{-1}\left(\frac{U(s)}{\Lambda(s)}\right)$; similarly, let $\ell_t=\mathcal{L}^{-1}\left(\frac{Y(s)}{\Lambda(s)}\right)$, and then
	$$y_t=a_0\ell_t+a_1\ell_t^{(1)}+\cdots+a_{n-1}\ell_t^{n-1}+\ell_t^{(n)}.$$
	Thus, $[z_t^\top,\ell_t^\top]^\top$ can be regarded as an internal latent variable that generates the input-output of system (\ref{system}), which aligns with the conclusion in \cite[Section III.A]{In-out}. However, in contrast to \cite{In-out}, we explicitly derive the specific form and physical meaning of $[z_t^\top,\ell_t^\top]^\top$, which is no longer merely a symbolic quantity.
    Furthermore, we show that the jet vector formed from $[z_t^\top,\ell_t^\top]^\top$, namely
	$[(g_t^u)^\top,(g_t^y)^\top]^\top=[(z_t)^\top,\cdots,(z_t^{(n-1)})^\top,(\ell_t)^\top,\cdots,(\ell_t^{(n-1)})^\top]^\top$, spans a space of the same dimension as the input-output jet vector $[(u_t)^\top,\cdots,(u_t^{(n-1)})^\top,(y_t)^\top,\cdots,(y_t^{(n-1)})^\top]^\top$ employed in \cite{In-out}. Unlike the jet vector in \cite{In-out}, which requires high-order derivatives that must be approximated numerically, our jet vector $[(g_t^u)^\top,(g_t^y)^\top]^\top$ yields true values in a model-free manner via equation (\ref{eta-sys}) (noting that the permutation matrix $P^{\eta}$ is known).
    Additionally, 
    assertion (b) of Lemma \ref{le1} will serve to refine the continuous-time trajectory generation characterization of Willems’ Fundamental Lemma over that in \cite{In-out}.
    This topic will be analyzed in depth in Section \ref{sec5}.
\end{remark}

By Lemma \ref{le1}, for multi-output systems ($p>1$), the data matrix $E_0$ is not full row rank. Directly replacing the unknown full row rank true state matrix $X_0:=[x_t,x_{t+\Delta t},\cdots,x_{t+(T-1)\Delta t}]$ with a rank-deficient $E_0$ may fail to completely parameterize the closed-loop system and feedback controller. Furthermore, algorithms for solving Problem \ref{P1} may suffer from numerical difficulties, instability, or non-convergence. However, under sufficient PE conditions, the linearly independent rows of $E_0$ contain all the true system information. To ensure the resulting data matrix is full row rank, we perform a full QR decomposition on $E_0$ and use the first $(m+2)n$ rows of the transposed orthogonal matrix as a projection. Projecting $\eta_t$ via this matrix yields an improved substitute state $\phi_{t}$, whose corresponding data matrix is guaranteed to be full row rank and possesses an associated isosystem. This is formalized in Lemma \ref{le2}.

    \begin{lemma}\label{le2}
	If the data matrix $E_0:=[\eta_t,\cdots,\eta_{t+(T-1)\Delta t}]$ obtained from system (\ref{eta-sys}) undergoes a full QR decomposition such that $E_0=F[\Phi_0^\top,\mathbf{0}_{T\times (p-1)n}]^\top$ with orthogonal matrix $F\in\mathbb{R}^{(m+p+1)n\times (m+p+1)n}$, define $F_1$ as the matrix formed by the first $(m+2)n$ rows of $F^{\top}$, and let the improved substitute state be $\phi_t:=F_1\eta_{t}$. Then, the following assertions hold.
	\begin{itemize}
		\item[\textup{(a)}] $\phi_t$ satisfies the isosystem given by
		\begin{equation}\label{full-c-sys}
			\dot{\phi}_t=F_1(\mathcal{A}_s+\mathcal{B}_s^yCS)F_1^\top \phi_t+F_1\mathcal{B}_s^u u_t,\quad y_t=CSF_1^\top \phi_t.
		\end{equation}
		\item[\textup{(b)}] If the input to system (\ref{system}) satisfies the $2n$-th order PE condition, the space spanned by the improved substitute state $\phi_t$ is $\mathbb{R}^{(m+2)n}$, and the resulting data matrix $\Phi_0:=[\phi_t,\cdots,\phi_{t+(T-1)\Delta t}]$ is guaranteed to be full row rank.
	\end{itemize}
\end{lemma}
\begin{proof}     
	For assertion (a), let $F^{\top}=[F_1^\top,F_2^\top]^\top$ with $F_1\in\mathbb{R}^{(m+2)n\times (m+p+1)n}$ and $F_2\in\mathbb{R}^{(p-1)n\times (m+p+1)n}$. Then,
	$F=[F_1^\top,F_2^\top]$, and by the definitions of $F$ and $\phi_t$, we have $\eta_t=F[\phi_t^\top,\mathbf{0}]^\top=F_1^\top\phi_t$. Substituting this into equation (\ref{eta}) yields
	\begin{align*}
		\dot{\phi}_t&=F_1(\mathcal{A}_s+\mathcal{B}_s^yCS) \eta_t+F_1\mathcal{B}_s^u u_t=F_1(\mathcal{A}_s+\mathcal{B}_s^yCS) F[\phi_t^\top,\mathbf{0}]^\top+F_1\mathcal{B}_s^u u_t\\
		&=F_1(\mathcal{A}_s+\mathcal{B}_s^yCS) F_1^\top\phi_t+F_1\mathcal{B}_s^u u_t,
	\end{align*}
	and a similar derivation yields the equation for $y_t$.
	
	For assertion (b), 
	from assertion (b) of Lemma \ref{le1} and the relation $F^\top \eta_t=[\phi_t^\top,\mathbf{0}_{1\times (p-1)n}]^\top$, it follows that the space spanned by $\eta_t$ is contained in the subspace spanned by the first $(m+2)n$ columns of $F$ (i.e., $F_1$), with $\phi_t$ serving as the coordinate representation within this subspace. Consequently, under the $2n$-th order PE input, the space spanned by $\phi_t$ coincides with that of $\eta_t$, which is $\mathbb{R}^{(m+2)n}$. This further implies that $\Phi_0$ has full row rank.
\end{proof}

Under the projection matrix $F_1$, it follows that $x_t=SF_1^\top\phi_t$. Now we verify if the improved state parameterization matrix $SF_1^\top$ remains full row rank. Since $X_0=SE_0=S[F_1^\top,F_2^\top][\Phi_0^\top,\mathbf{0}]^\top=(SF_1^\top) \Phi_0$, and both $X_0$ and $\Phi_0$ are full row rank under the $2n$-th order PE condition, the improved state parameterization matrix $SF_1^\top$ is guaranteed to be full row rank.

 Define $U_0:=[u_{t},\cdots,u_{t+(T-1)\Delta t}]$. Theorem \ref{thm2} demonstrates the role of the full row rank substitute state data matrix $\Phi_0$ in the parameterization of the closed-loop systems and feedback control gains.
\begin{theorem}\label{thm2}
	If the input to system (\ref{system}) satisfies the $(2n+1)$-th order PE condition, the full row rank matrix $\Phi_0$ constructed according to Lemma \ref{le2} can be used to completely parameterize the closed-loop form of system (\ref{full-c-sys}) and the feedback controller $u_t=K\phi_t$, i.e.,
	\begin{equation}\label{para}
		\mathrm{im}\left(\begin{bmatrix}
			F_1\\K
		\end{bmatrix}\right)\subseteq\mathrm{im}\left(\begin{bmatrix}
		\Phi_0\\U_0
	\end{bmatrix}\right).
	\end{equation}
\end{theorem}
The proof of this theorem follows similarly to the discrete-time case in \cite{DD-PIVI}. Furthermore, extending the discrete-time results confirms that for multi-output systems, $E_0$ cannot yield a representation analogous to (\ref{para}). Beyond its role in Theorem \ref{thm2}, if $\Phi_0$ is placed on the far right (or $\Phi_0^\top$ on the far left) in the model-free iteration equations, post-multiplying (or pre-multiplying) by its right inverse (or the transpose thereof) yields equations equivalent to the original ones. This facilitates the design of data-enabled algorithms, as will be demonstrated in Section \ref{sec4}.

\subsection{Equivalent LQR Problem}\label{sec3_2}
This subsection transforms the information-limited Problem~\ref{P1} into an equivalent new LQR problem with fewer unknowns. Using $x_t=SF_1^\top \phi_t$, the original cost functional (\ref{J}) can be rewritten as
\begin{equation}\label{full-J}
	J(\phi_0,K)=\int_{0}^{\infty}(\phi^\top_t Q_{\phi}\phi_{t}+ u_t^\top R u_t)dt,
\end{equation}
where $Q_{\phi}:=F_1S^\top C^\top QCSF_1^\top\in\mathbb{S}^{(m+2)n}_{+}$ is unknown, even though its partial structure $Q$ is known. Furthermore, since the improved substitute state $\phi_t$ and input $u_t$ satisfy the linear system (\ref{full-c-sys}), Problem \ref{P1} can be reformulated as Problem \ref{P2}.
\begin{problem}\label{P2}
	Consider the LQR problem defined by system (\ref{full-c-sys}) and cost functional (\ref{full-J}), where the dynamic matrix $F_1(\mathcal{A}_s+\mathcal{B}_s^yCS)F_1^\top$ and the weighting matrix $Q_{\phi}$ are unknown, while the partial structure $Q$ of $Q_{\phi}$, the input matrix $F_1\mathcal{B}_s^u$, the weighting matrix $R$, and the state $\phi_t$ are known. The objective is to find an optimal feedback control gain $K$ that minimizes the cost functional (\ref{full-J}) under the control policy $u_t=-K\phi_{t}$.
\end{problem}

Now we prove that Problem \ref{P1} is equivalent to Problem \ref{P2}. We first analyze the relationship between the parameter matrices of system (\ref{system}) and system (\ref{full-c-sys}), which is summarized in Lemma \ref{le3}.

\begin{lemma}\label{le3}
	Consider systems (\ref{system}) and (\ref{full-c-sys}). The following relationships hold
	\begin{equation}\label{AM}
		SF_1^\top F_1(\mathcal{A}_s+\mathcal{B}_s^y CS)F_1^\top=ASF_1^\top,\quad SF_1^\top F_1\mathcal{B}_s^u=B.
	\end{equation}
\end{lemma}

 \begin{proof}
	We first prove the second equality. From the definitions of $\mathcal{D}$, $S$, $S^u$, $S^u_i$ for $i=1,\cdots,m$, and $\mathcal{B}_s^{u}=[(I_{m}\otimes b_s)^\top, \mathbf{0}_{m\times (p+1)n}]^\top$ in Section \ref{sec3_1}, together with $D_{n-1}=I_n$, it holds that
	\begin{align}\label{BB}
		\nonumber S\mathcal{B}_s^u&=S^u(I_m\otimes b_s)=[S_1^ub_s,\cdots,S_m^ub_s]\\\nonumber
		&=[[D_0B_1,\cdots,D_{n-1}B_1]b_s,\cdots,[D_0B_m,\cdots,D_{n-1}B_m]b_s]\\\nonumber
		&=[D_{n-1}B_1,\cdots,D_{n-1}B_m]\\
		&=[B_1,\cdots,B_m]=B.
	\end{align}   
    Since $\phi_t=F_1\eta_t$ and $\eta_t=[F_1^\top,F_2^\top][\phi_t^\top,\mathbf{0}]^\top=F_1^\top\phi_t$, it follows that $F_1^\top F_1\eta_t=F_1^\top\phi_t=\eta_t$. Thus, $\eta_t$ always lies in the subspace $\mathrm{im}(F_1^\top F_1)$.
    Assume that isosystem (\ref{C-sys}) is initially relaxed. Applying a unit impulse to the $i$-th channel at $t=0$, i.e., $u_t=\delta_t e_i$ for $i=1,\cdots,m$, where $\delta_t$ is the Dirac delta function and $e_i$ is the $i$-th standard basis vector of $\mathbb{R}^m$, the zero-state response for any $t>0$ is $\eta_t=e^{(\mathcal{A}_s+\mathcal{B}_s^yCS)t}\mathcal{B}_s^ue_i\in\mathcal{V}$. Furthermore, since the span $\mathcal{V}$ of $\eta_{t}$ is a linear subspace and is closed under limits, it follows that
	$\lim\limits_{t\rightarrow 0_{+}}e^{(\mathcal{A}_s+\mathcal{B}_s^yCS)t}\mathcal{B}_s^ue_i=\mathcal{B}_s^ue_i\in\mathcal{V}$.
	Thus, $F_1^\top F_1\mathcal{B}_s^u=\mathcal{B}_s^u$. Premultiplying both sides by $S$ yields
	$SF_1^\top F_1\mathcal{B}_s^u=S\mathcal{B}_s^u=B$, which completes the proof of the second equality.
	
	For the first equality, following a derivation similar to that for (\ref{BB}), we obtain
	$S\mathcal{B}_y^u=L$.
	From the definition of $\mathcal{D}$ and the relations satisfied by its block matrices $D_{i}$,  namely $(A-LC)D_0=-a_0I_n$, $(A-LC)D_{i+1}=D_{i}-a_{i+1}D_{n-1}$ for $i=0,\cdots,n-3$, and $D_{n-1}=I_n$, it follows that for $i=1,\cdots,m$,
	\begin{align*}
		S_i^u A_s&=\mathcal{D}(I_n\otimes B_i)A_s=[D_0B_i,\cdots,D_{n-1}B_i]\mathcal{A}_s\\
		&=[-a_{0}D_{n-1},D_0-a_1D_{n-1},\cdots,D_{n-2}-a_{n-1}D_{n-1}](I_n\otimes B_i)\\
		&=(A-LC)[D_0,\cdots,D_{n-1}](I_n\otimes B_i)=(A-LC)S_i^u,
	\end{align*}
	Similarly, for $j=1,\cdots,p$, we have $S_j^y A_s=(A-LC)S_j^y$.
	Note that for $A-LC$, it follows from the relationship between the determinant and the adjugate matrix, the definition of $\mathcal{D}$, and $D_{n-1}=I_n$ that $(A-LC)\mathcal{D}=\mathcal{D}A_s$. Since $A-LC$ and $A_{\epsilon}$ share the same characteristic polynomial, replacing $A-LC$ with $A_{\epsilon}$ and $\mathcal{D}$ with $\mathcal{D}_{\epsilon}$ in the above analysis similarly yields $A_{\epsilon}\mathcal{D}_{\epsilon}=\mathcal{D}_{\epsilon}A_s$.
	Considering the coupling between the coefficient matrix of the adjugate  matrix and the initial state, it holds readily that
	$$(A-LC)\mathcal{D}(I_n\otimes \epsilon_0)=\mathcal{D}(I_n\otimes \epsilon_0)A_s,\quad A_{\epsilon}\mathcal{D}_{\epsilon}(I_n\otimes \eta_0^{\epsilon})=\mathcal{D}_{\epsilon}(I_n\otimes \eta_0^{\epsilon})A_s.$$
	Therefore, from the invertibility of $\mathcal{D}_{\epsilon}(I_n\otimes \eta_0^{\epsilon})$, it follows that
	\begin{align*}
		(A-LC)S^{\epsilon}&=(A-LC)\mathcal{D}(I_n\otimes \epsilon_0)(\mathcal{D}_{\epsilon}(I_n\otimes \eta_0^{\epsilon}))^{-1}\\
		&=\mathcal{D}(I_n\otimes \epsilon_0)A_s(\mathcal{D}_{\epsilon}(I_n\otimes \eta_0^{\epsilon}))^{-1}\\
		&=\mathcal{D}(I_n\otimes \epsilon_0)(\mathcal{D}_{\epsilon}(I_n\otimes \eta_0^{\epsilon}))^{-1}A_{\epsilon}\\
		&=S^{\epsilon}A_{\epsilon}.
	\end{align*}
    Hence,
	\begin{align*}
		S\mathcal{A}_s&=[S^u(I_m\otimes A_s),S^y(I_p\otimes A_s),S^{\epsilon}A_{\epsilon}]\\
		&=[(A-LC)S^u,(A-LC)S^y,(A-LC)S^{\epsilon}]=(A-LC)S.
	\end{align*}
	Synthesizing the above results, we have
	\begin{equation}\label{SA}
		S(\mathcal{A}_s+\mathcal{B}_s^yCS)=AS.
	\end{equation}
	Postmultiplying both sides by $F_1^\top$ yields $S(\mathcal{A}_s+\mathcal{B}_s^yCS)F_1^\top=ASF_1^\top$. The zero-input response of system (\ref{C-sys}) is $\eta_t=e^{(\mathcal{A}_s+\mathcal{B}_s^yCS)t}\eta_0\in\mathcal{V}$. If for any $v\in\mathcal{V}$, we set $\eta_0=v$, then $e^{(\mathcal{A}_s+\mathcal{B}_s^yCS)t}v\in\mathcal{V}$, implying that $\mathcal{V}$ is invariant under the flow $e^{(\mathcal{A}_s+\mathcal{B}_s^yCS)t}$. Furthermore, since $\mathcal{V}$ is a finite-dimensional linear subspace and is closed under differentiation, it follows that $\frac{d (e^{(\mathcal{A}_s+\mathcal{B}_s^yCS)t}v)}{dt}\big\vert_{t=0}=(\mathcal{A}_s+\mathcal{B}_s^yCS)v\in\mathcal{V}$, implying that the span $\mathcal{V}$ of $\eta_{t}$ is an $(\mathcal{A}_s+\mathcal{B}_s^yCS)$-invariant subspace. Furthermore, since $\mathcal{V}=\mathrm{im}(F_1^\top)$, for any $v\in\mathcal{V}$, there exists $w\in\mathbb{R}^{(m+2)n}$ such that $v=F_1^\top w$. By the $(\mathcal{A}_s+\mathcal{B}_s^yCS)$-invariance of $\mathcal{V}$, there exists a matrix	$M\in\mathbb{R}^{(m+2)n\times (m+2)n}$ such that
	$$(\mathcal{A}_s+\mathcal{B}_s^yCS)v=(\mathcal{A}_s+\mathcal{B}_s^yCS)F_1^\top w=F_1^\top(Mw),$$
	By the arbitrariness of $v$, it follows that $(\mathcal{A}_s+\mathcal{B}_s^yCS)F_1^\top =F_1^\top M$. Premultiplying both sides by $F_1^\top F_1$ yields
	$$F_1^\top F_1(\mathcal{A}_s+\mathcal{B}_s^yCS)F_1^\top=F_1^\top F_1 F_1^\top M=F_1^\top M=(\mathcal{A}_s+\mathcal{B}_s^yCS)F_1^\top.$$
	Premultiplying both sides by $S$ yields  $SF_1^\top F_1(\mathcal{A}_s+\mathcal{B}_s^yCS)F_1^\top=S(\mathcal{A}_s+\mathcal{B}_s^yCS)F_1^\top=ASF_1^\top$, which proves the first equality.
	
\end{proof}

The ARE satisfied by the LQR problem consisting of system (\ref{full-c-sys}) and cost functional (\ref{full-J}) is given by
$$F_1(\mathcal{A}_s+\mathcal{B}_s^yCS)^\top F_1^\top \Sigma+\Sigma F_1(\mathcal{A}_s+\mathcal{B}_s^yCS)F_1^\top+F_1S^\top C^\top QCSF_1^\top-\Sigma F_1\mathcal{B}_s^uR^{-1}(F_1\mathcal{B}_s^u)^\top \Sigma=\mathbf{0}.$$
Let $P_x=(F_1S^\top)^\dagger\Sigma(SF_1^\top)^\dagger $. Using (\ref{AM}), the above equation can be rewritten as
$$(SF_1^\top)^\top(A^\top P_x+P_xA+Q_x-P_xBR^{-1}B^\top P_x)SF_1^\top=\mathbf{0},$$
which is equivalent to the state feedback ARE (\ref{ARE}) by the full row rank property of $SF_1^\top$. Similarly, the optimal feedback gain for the LQR problem consisting of system (\ref{full-c-sys}) and cost functional (\ref{full-J}) is given by
$$K=R^{-1}(F_1\mathcal{B}_s^{u})^\top \Sigma=R^{-1}BP_xSF_1^\top=K_xSF_1^\top.$$
The Lyapunov equations corresponding to the two LQR problems satisfy a similar equivalence relation. Therefore, by the full row rank property of
$SF_1^\top$, the feedback controller $u_t=K\phi_t$ is performance equivalent to the state feedback controller $u_t=K_xx_t$; by the equivalence of the AREs, the optimal solution $u_t=K^*\phi_t$ to Problem \ref{P2} must be the optimal solution to Problem \ref{P1}.

\begin{remark}
	The work in \cite{Huang-newLQR} also converts the original LQR problem with unknown dynamics and unmeasurable states into an equivalent LQR problem with known input matrix and accessible states. However, in contrast to \cite{Huang-newLQR}, the new LQR problem constructed in this paper is strictly equivalent to the original problem, rather than being equivalent only in the steady-state sense. This is primarily reflected in the following aspects: by accounting for the replication of the observation error $\epsilon_t$ when constructing the substitute state, the linear injective relationship between the substitute state and the true state holds strictly, and the equivalence of the cost functionals is also strictly valid; moreover, the isosystem (\ref{full-c-sys}) in this paper is strictly equivalent to the original system (\ref{system}) without neglecting the unknown observation error $\epsilon_t$. In addition, a significant improvement is the use of the projection matrix $F_1$ to further eliminate redundant information from the input-output filtering vector $\eta_t$, which not only facilitates the design of efficient model-free algorithms but also reduces computational complexity, as will be demonstrated in Section \ref{sec4}.
\end{remark}

	\section{Efficient Off-Policy Value-Based Iteration Algorithms}\label{sec4}
	To design model-free algorithms for solving Problem \ref{P2}, this section first analyzes the stabilizability and detectability of the LQR problem defined by system (\ref{full-c-sys}) and cost functional (\ref{full-J}). Then, based on model-based results, efficient PI and VI algorithms are developed by utilizing the full row rank data matrices $\Phi_0$, $U_0$, $Y_0:=[y_{t},\cdots,y_{t+(T-1)\Delta t}]$, and $\Phi_1:=[\dot{\phi}_t,\cdots,\dot{\phi}_{t+(T-1)\Delta t}]=F_1[\dot{\eta}_t,\cdots,\dot{\eta}_{t+(T-1)\Delta t}]$ obtained from Lemma \ref{le2} under $2n$-th order PE inputs.
	
	\subsection{Regularity Condition}\label{sec4_1}
	\begin{lemma}\label{le4}
		Under Assumption \ref{as1}, if all eigenvalues of $A-LC$ are chosen to be in the open left-half plane, then the pair $(F_1(\mathcal{A}_s+\mathcal{B}_s^yCS)F_1^\top ,F_1\mathcal{B}_s^u)$ is stabilizable.
	\end{lemma}

    \begin{proof}
    	We first prove that the pair $(\mathcal{A}_s+\mathcal{B}_s^yCS,\mathcal{B}_s^u)$ is stabilizable before projection. From the controllability assumption of system (\ref{system}), there exists a state feedback control gain $K_x$ such that $A-BK_x$ is Hurwitz stable. Applying $u_t=-K\eta_t=-K_xS\eta_t$ to system (\ref{system}) yields
    	$$\dot{x}_t=Ax_t+Bu_t=Ax_t-BK_xS\eta_t=(A-BK_x)x_t.$$
    	Hence, $x_t$ converges exponentially to $\mathbf{0}$. Considering the isosystem (\ref{C-sys}), applying $u_t=-K\eta_t=-K_xS\eta_t$ to it yields
    	\begin{align*}
    		\dot{\eta}_t&=(\mathcal{A}_s+\mathcal{B}_s^yCS)\eta_t-\mathcal{B}_s^uK\eta_t=(\mathcal{A}_s+\mathcal{B}_s^yCS-\mathcal{B}_s^uK)\eta_t,\\
    		\dot{\eta}_t&=(\mathcal{A}_s+\mathcal{B}_s^yCS)\eta_t-\mathcal{B}_s^uK_xS\eta_t=\mathcal{A}_s\eta_t+(\mathcal{B}_s^yC-\mathcal{B}_s^uK_x)x_t.
    	\end{align*}
    	Since $x_t$ converges exponentially to zero and $A-LC$ is Hurwitz stable when the eigenvalues of $\mathcal{A}_s$ are chosen to be in the open left-half plane, it follows from \cite{Huang-exp} that $\eta_t$ also converges exponentially to $\mathbf{0}$. Thus, the closed-loop system matrix $\mathcal{A}_s+\mathcal{B}_s^yCS-\mathcal{B}_s^uK$ must be Hurwitz, which proves that the pair $(\mathcal{A}_s+\mathcal{B}_s^yCS,\mathcal{B}_s^u)$ is stabilizable. If the orthogonal matrix $F^\top=[F_1^\top,F_2^\top]^\top$ defined in Lemma \ref{le2} is used to perform a similarity transformation on system (\ref{C-sys}), where the state of the transformed system is $\mu_t=F^\top\eta_t$, then by the properties of similarity transformations, the pair $(F^\top(\mathcal{A}_s+\mathcal{B}_s^yCS)F,F^\top\mathcal{B}_s^u)$ is also stabilizable.
    	
    	Building on this, we next use the Hautus stabilizability criterion to prove that the pair $(F_1(\mathcal{A}_s+\mathcal{B}_s^yCS)F_1^\top ,F_1\mathcal{B}_s^u)$ is stabilizable, i.e., for all $\lambda$ with positive real parts, the matrix $[\lambda I_{(m+2)n}-F_1(\mathcal{A}_s+\mathcal{B}_s^yCS)F_1^\top,F_1\mathcal{B}_s^u]$ has full row rank. For brevity, denote $\mathcal{A}:=\mathcal{A}_s+\mathcal{B}_s^yCS$, so we have
    	\begin{equation}\label{FAF}
    		F^\top\mathcal{A}F=\begin{bmatrix}
    			F_1 \mathcal{A} F_1^\top & F_1\mathcal{A} F_2^\top\\ \mathbf{0}&F_2\mathcal{A}F_2^\top
    		\end{bmatrix},\quad F^\top\mathcal{B}_s^u=\begin{bmatrix}
    			F_1\mathcal{B}_s^u\\\mathbf{0}
    		\end{bmatrix},
    	\end{equation}
    	where the zero matrices arise because, the $\mathcal{A}$-invariance of the span $\mathcal{V}=\mathrm{im}(F_1^\top)$ of $\eta_t$ (established in the proof of Lemma \ref{le3}) implies $\mathcal{A}F_1^\top\in\mathcal{V}$, and thus $F_2\mathcal{A}F_1^\top=\mathbf{0}$ by the orthogonality of $F_2$ and $F_1^\top$; furthermore, since every column of $\mathcal{B}_s^u$ belongs to $\mathcal{V}=\mathrm{im}(F_1^\top)$ (also from the proof of Lemma \ref{le3}), it follows that $F_2\mathcal{B}_s^u=\mathbf{0}$. Therefore, we have
    	$$[\lambda I_{(m+p+1)n}-F^\top\mathcal{A}F,F^\top\mathcal{B}_s^u]=\begin{bmatrix}
    		\lambda I_{(m+2)n}- F_1 \mathcal{A} F_1^\top & -F_1\mathcal{A} F_2^\top& F_1\mathcal{B}_s^u\\ \mathbf{0}&\lambda I_{(p-1)n}- F_2\mathcal{A}F_2^\top&\mathbf{0}
    	\end{bmatrix},$$
    	Then,
    	$$\mathrm{rank}([\lambda I_{(m+p+1)n}-F^\top\mathcal{A}F,F^\top\mathcal{B}_s^u])=\mathrm{rank}([	\lambda I_{(m+2)n}- F_1 \mathcal{A} F_1^\top, F_1\mathcal{B}_s^u])+\mathrm{rank}(\lambda I_{(p-1)n}- F_2\mathcal{A}F_2^\top).$$
    	For the system after the similarity transformation, let $\mu_t=F^\top\eta_t=[(\mu_t^1)^\top,(\mu_t^2)^\top]^\top$ with $\mu_t^1\in\mathbb{R}^{(m+2)n}$ and $\mu_t^2\in\mathbb{R}^{(p-1)n}$. Then $\dot{\mu}_t^2=F_2\mathcal{A}F_2^\top\mu_t^2$ forms an autonomous subsystem. Since the pair $(F^\top(\mathcal{A}_s+\mathcal{B}_s^yCS)F,F^\top\mathcal{B}_s^u)$ is stabilizable, the subsystem $\dot{\mu}_t^2=F_2\mathcal{A}F_2^\top\mu_t^2$ cannot contain unstable modes, implying that $F_2\mathcal{A}F_2^\top$ is Hurwitz, i.e., $\mathrm{rank}(\lambda I_{(p-1)n}- F_2\mathcal{A}F_2^\top)=(p-1)n$ for all $\lambda$ with positive real parts. Then it holds that
    	$\mathrm{rank}([	\lambda I_{(m+2)n}- F_1 \mathcal{A} F_1^\top, F_1\mathcal{B}_s^u])=(m+2)n$ for all $\lambda$ with positive real parts, that is, $(F_1(\mathcal{A}_s+\mathcal{B}_s^yCS)F_1^\top ,F_1\mathcal{B}_s^u)$ is stabilizable.
    \end{proof}
	
	\begin{lemma}\label{le5}
		Under Assumption \ref{as1}, if all eigenvalues of $A-LC$ are chosen to be in the open left-half plane, then the pair $(F_1(\mathcal{A}_s+\mathcal{B}_s^yCS)F_1^\top,Q^{\frac{1}{2}}CSF_1^\top)$ is detectable.
	\end{lemma}
	
	 \begin{proof}
		First, under Assumption \ref{as1}, we use the Hautus detectability criterion prove that the pair $(\mathcal{A}_s+\mathcal{B}_s^yCS,Q^{\frac{1}{2}}CS)$ is detectable before projection. Using proof by contradiction, suppose that for some $\lambda$ with positive real part, there exists a non-zero vector $v\in\mathbb{R}^{(m+p+1)n}$ such that
		$$\begin{bmatrix}
			\lambda I_{(m+p+1)n}-(\mathcal{A}_s+\mathcal{B}_s^yCS)\\Q^{\frac{1}{2}}CS
		\end{bmatrix}v=\mathbf{0}.$$
		From equation (\ref{SA}) in the proof of Lemma \ref{le3}, we have $S(\mathcal{A}_s+\mathcal{B}_s^yCS)=AS$. Post-multiplying both sides of this equation by the vector $v$ yields $S(\mathcal{A}_s+\mathcal{B}_s^yCS)v=ASv$. According to the assumption, $S(\mathcal{A}_s+\mathcal{B}_s^yCS)v=A(Sv)=\lambda (Sv)$ which implies that $\lambda$ is an eigenvalue of matrix $A$ associated with the eigenvector $Sv$. Furthermore, since $Q^{\frac{1}{2}}C(Sv)=\mathbf{0}$, the observability of $(A,Q^{\frac{1}{2}}C)$ in Assumption \ref{as1} implies $Sv=\mathbf{0}$. Consequently, $(\mathcal{A}_s+\mathcal{B}_s^yCS)v=\mathcal{A}_sv=\lambda v$, meaning that $\lambda$ is an eigenvalue of  $\mathcal{A}_s$ associated with the eigenvector $v$. However, since all eigenvalues of $A-LC$ are chosen to be in the open left-half plane, $\mathcal{A}_s$ is Hurwitz stable, and thus all its eigenvalues must have negative real parts, leading to a contradiction. Therefore, the pair $(\mathcal{A}_s+\mathcal{B}_s^yCS,Q^{\frac{1}{2}}CS)$ is detectable. It follows that the pair $(F^\top(\mathcal{A}_s+\mathcal{B}_s^yCS)F,Q^{\frac{1}{2}}CSF)$ is also detectable under the similarity transformation $\mu_t=F^\top\eta_t$.

		Building on this, we next use the Hautus detectability criterion and proof by contradiction to show that the pair $(F_1(\mathcal{A}_s+\mathcal{B}_s^yCS)F_1^\top,Q^{\frac{1}{2}}CSF_1^\top)$ is detectable. Suppose that for some $\lambda$ with positive real part, there exists a non-zero vector $w\in\mathbb{R}^{(m+2)n}$ such that
		$$\begin{bmatrix}
			\lambda I_{(m+2)n}-F_1(\mathcal{A}_s+\mathcal{B}_s^yCS)F_1^\top\\Q^{\frac{1}{2}}CSF_1^\top
		\end{bmatrix}w=\mathbf{0}.$$
		 For brevity, denote $\mathcal{A}:=\mathcal{A}_s+\mathcal{B}_s^yCS$. Since $Q^{\frac{1}{2}}CSF=[Q^{\frac{1}{2}}CSF_1^\top,Q^{\frac{1}{2}}CSF_2^\top]$, combining this with equation (\ref{FAF}) yields
		$$\begin{bmatrix}
			\lambda I_{(m+p+1)n}-F^\top\mathcal{A}F\\Q^{\frac{1}{2}}CSF
		\end{bmatrix}=\begin{bmatrix}
			\lambda I_{(m+2)n}-F_1\mathcal{A}F_1^\top& -F_1\mathcal{A}F_2^\top\\
			\mathbf{0}&\lambda I_{(p-1)n}-F_2\mathcal{A}F_2^\top\\
			Q^{\frac{1}{2}}CSF_1^\top&Q^{\frac{1}{2}}CSF_2^\top
		\end{bmatrix}.$$
		Thus,
		$$\begin{bmatrix}
			\lambda I_{(m+p+1)n}-F^\top\mathcal{A}F\\Q^{\frac{1}{2}}CSF
		\end{bmatrix}\begin{bmatrix}
			w\\\mathbf{0}_{(p-1)n}
		\end{bmatrix}=\begin{bmatrix}
			\lambda I_{(m+2)n}-F_1\mathcal{A}F_1^\top\\
			\mathbf{0}\\
			Q^{\frac{1}{2}}CSF_1^\top
		\end{bmatrix}w=\mathbf{0}.$$
	    This contradicts the detectability of the pair $(F^\top(\mathcal{A}_s+\mathcal{B}_s^yCS)F,Q^{\frac{1}{2}}CSF)$, which proves that the pair $(F_1(\mathcal{A}_s+\mathcal{B}_s^yCS)F_1^\top,Q^{\frac{1}{2}}CSF_1^\top)$ is detectable.
	\end{proof}

    \begin{remark}
    	In analyzing the stabilizability and detectability of the unprojected LQR problem, our approach differs from that in \cite{Huang-newLQR} in two key aspects. First, we do not need to consider the influence of the observation error $\epsilon_t$ when establishing stabilizability. Second, regarding detectability, while \cite{Huang-newLQR} requires $Q\succ0$ and relies on constructing a specific observation gain, we provide a rigorous proof via the Hautus detectability criterion; by leveraging the observability of the original system and the matrix relationship between the isosystem and the original system, our method only requires $Q\succeq0$.
    \end{remark}

    Consequently, based on the results in the appendix of \cite{Huang-newLQR}, the stabilizability and detectability of the transformed new LQR problem guarantee the existence of convergent model-based iterative algorithms. This provides the theoretical foundation for designing model-free algorithms to solve Problem \ref{P2}.

   \subsection{Policy Iteration Algorithm}\label{sec4_2}
   Considering the equivalent transformed LQR problem, define $\mathcal{A}_F:=F_1(\mathcal{A}_s+\mathcal{B}_s^yCS)F_1^\top$ and $\mathcal{B}_F:=F_1\mathcal{B}_s^u$. Then, under the feedback control $u_t=-K\phi_t$, the Lyapunov equation satisfied by this LQR problem is
   \begin{equation}\label{full-lyap}
   	(\mathcal{A}_F-\mathcal{B}_FK)^\top \Sigma+\Sigma (\mathcal{A}_F-\mathcal{B}_FK)+Q_{\phi}+K^\top RK=\mathbf{0},
   \end{equation}
   which admits a unique optimal solution $\Sigma\in\mathbb{S}^{(m+2)n}_{+}$ under Lemmas \ref{le4} and \ref{le5}. Furthermore, considering the cost functional (\ref{full-J}), the optimal feedback control gain is
   \begin{equation}\label{full-K}
   	K=R^{-1}\mathcal{B}_F^{\top}\Sigma.
   \end{equation}
   Considering the setup of Problem \ref{P2}, since $C$ and $S$ are unknown matrices, $\mathcal{A}_F$ and portions of the structure of $Q_{\phi}$ are unknown. Nevertheless, since $\mathcal{B}_F$ is known, we can left-multiply both sides of equation (\ref{full-lyap}) by $\Phi_0^\top$ and right-multiply by $\Phi_0$, then add and subtract the term $U_0^\top \mathcal{B}_F^\top \Sigma\Phi_0+\Phi_0^\top \Sigma\mathcal{B}_FU_0$. Rearranging yields
   \begin{align*}
   	&(\mathcal{A}_F\Phi_0+\mathcal{B}_FU_0-\mathcal{B}_F(K\Phi_0+U_0))^\top\Sigma\Phi_0
   	+\Phi_0^\top\Sigma (\mathcal{A}_F\Phi_0+\mathcal{B}_FU_0-\mathcal{B}_F(K\Phi_0+U_0))\\
   	&+\Phi^\top Q_{\phi}\Phi_0+\Phi_0^\top K^\top RK\Phi_0=\mathbf{0}.
   \end{align*}
   Noting that $\Phi_1=\mathcal{A}_F\Phi_0+\mathcal{B}_FU_0$ and $\Phi_0^\top Q_{\phi}\Phi_0=\Phi_0^\top F_1S^\top C^\top QCSF_1^\top\Phi_0=X_0^\top C^\top QCX_0=Y_0^\top Q Y_0 $, the above equation simplifies to
   \begin{equation}\label{PI-eval}
   	(\Phi_1-\mathcal{B}_F(K\Phi_0+U_0))^\top\Sigma\Phi_0
   	+\Phi_0^\top\Sigma (\Phi_1-\mathcal{B}_F(K\Phi_0+U_0))
   	+Y_0^\top QY_0+\Phi_0^\top K^\top RK\Phi_0=\mathbf{0}.
   \end{equation}
   If the current feedback control gain $K$ is known, the only unknown in equation (\ref{PI-eval}) is $\Sigma$. Thus, equation (\ref{PI-eval}) can serve as the policy evaluation step in model-free PI. Furthermore, since $\mathcal{B}_F$ is known, equation (\ref{full-K}) can still be used as the policy improvement step. The complete algorithm is presented in Algorithm \ref{al1}. Since equation (\ref{PI-eval}) is a generalized Sylvester equation (if $T=(m+2)n$ is chosen, it reduces to a Lyapunov equation due to the full row rank property of $\Phi_0$), it can be solved directly, eliminating the need to convert the iterative equation into an LS problem for numerical solution or to consider the solvability of such an LS problem. Since Algorithm \ref{al1} only requires collecting data once prior to iteration and uses the same data for every iteration, the algorithm is off-policy.
   
   \begin{algorithm}
   	\caption{Model-Free Policy Iteration (PI)}
   	\begin{algorithmic}
   		\State \textbf{Data Collection:} Apply a $2n$-th order PE input to system (\ref{system}), and collect the corresponding input $u_t$ and output $y_t$ trajectories to form $U_0$ and $Y_0$. Concatenate $u_t$, $y_t$, and a vector of zeros into an extended vector, and use it as the input for system (\ref{eta-sys}) to collect the corresponding state $\eta_t$ and state derivative $\dot{\eta}_t$, forming $E_0$ and $E_1$. Calculate the projection matrix $F_1$, and construct the full row rank data matrices $\Phi_0=F_1E_0$, $\Phi_1=F_1E_1$, and $\mathcal{B}_F=F_1\mathcal{B}_s$.
   		\State \textbf{Initialization:} Select a precision constant $\epsilon$. Set $i=0$ and choose an initial stabilizing control gain $K^0$ such that $\mathcal{A}_F-\mathcal{B}_FK^0$ is Hurwitz stable.
   		\For{$i=0,1,2,\cdots,$}
   		\State (i) Policy evaluation: calculate $\Sigma^i$ by solving
   		\begin{equation}\label{PIi}
   			\begin{split}
   				&(\Phi_1-\mathcal{B}_F(K^i\Phi_0+U_0))^\top\Sigma^i\Phi_0
   				+\Phi_0^\top\Sigma^i (\Phi_1-\mathcal{B}_F(K^i\Phi_0+U_0))\\
   				&+Y_0^\top QY_0+\Phi_0^\top (K^i)^\top RK^i\Phi_0=\mathbf{0}.
   			\end{split}
   		\end{equation}
   		\State (ii) Policy improvement: update the control gain as
   		$K^{i+1}=R^{-1}\mathcal{B}_F^\top \Sigma^i$.
   		\If{$\Vert K^{i+1}-K^{i}\Vert\leq\epsilon$}
   		\State~~~~\Return $K_{PI}^*=K^{i+1}$.
   		\EndIf
   		\EndFor
   	\end{algorithmic}
   	\label{al1}
   \end{algorithm}

    \begin{remark}[Initial Stabilizing Control Gain]
    Since the states of system (\ref{full-c-sys}) are available, we can utilize continuous-time state feedback stabilization results to provide an initial stabilizing gain for Algorithm \ref{al1}. Two primary methods are presented below. \textup{(a)} Dead-beat method \cite{CT-DD-PI}. Let $\mathcal{G}^1$ and $\mathcal{G}^2$ be the right inverse and a basis matrix for the null space of $\Phi_0$, respectively. Design $K^{db}$ such that $\Phi_1(\mathcal{G}^1-\mathcal{G}^2K^{db})$ is Hurwitz stable. Then $K^0=-U_0(\mathcal{G}^1-\mathcal{G}^2K^{db})$ serves as an initial stabilizing control gain for Algorithm \ref{al1}. 
   	\textup{(b)} Linear matrix inequalities method \cite{DD-LMIs}. If there exists a matrix $\Gamma\in\mathbb{R}^{T\times(m+2)n}$ such that $\Phi_0 \Gamma\succ0$ and $\Phi_1\Gamma+\Gamma\Phi_1^\top\prec 0$, then $K^0=-U_0\Gamma(\Phi_0\Gamma)^{-1}$ serves as an initial stabilizing control gain for Algorithm \ref{al1}.
   \end{remark}

   \begin{proposition}
   	Under Assumption \ref{as1}, select the eigenvalues of $A-LC$ to lie in the open left-half plane. If the initial control gain $K^0$ is stabilizing, the sequence of control gains $\{K^i\}$ generated by Algorithm \ref{al1} is stabilizing, and the sequence of cost parameter matrices $\{\Sigma^i\}$ is monotonically non-increasing. Ultimately, both $\{\Sigma^i\}$ and $\{K^i\}$ converge to the optimal values $\Sigma^*$ and $K^*$, respectively.
   \end{proposition}
   \begin{proof}
   	By the full row rank property of $\Phi_0$ and the preceding analysis, iterative equation (\ref{PIi}) is equivalent to iterating on the model-based Lyapunov equation (\ref{full-lyap}). Thus, Algorithm \ref{al1} corresponds to Kleinman's algorithm \cite{model-based-PI} for the equivalent new LQR problem. Since \cite{Huang-newLQR} establishes that model-based PI remains stable and convergent under only stabilizability and detectability, Lemmas \ref{le4} and \ref{le5} imply that Algorithm \ref{al1} generates a sequence of stabilizing control gains and monotonically non-increasing cost parameter matrices, both converging to their respective optimal solutions.
   \end{proof}

   \subsection{Value Iteration Algorithm}\label{sec4_3}
   The VI method dispenses with the need for an initial stabilizing gain. Consider the following ARE equation associated with the consistent system (\ref{full-c-sys}) and cost functional (\ref{full-J}),
   $$\mathcal{A}_F^\top \Sigma +\Sigma \mathcal{A}_F+Q_{\phi}-\Sigma \mathcal{B}_F R^{-1}\mathcal{B}_F\Sigma=\mathbf{0}.$$
   If the current $\Sigma$ is known, $\Sigma \mathcal{B}_F R^{-1}\mathcal{B}_F\Sigma$ is computable since $\mathcal{B}_F$ is available, leaving only the first three terms to be addressed. Treating $\Omega:=\mathcal{A}_F^\top \Sigma +\Sigma \mathcal{A}_F+Q_{\phi}$ as a single unknown entity, we follow the same approach as in PI: premultiplying $\Omega$ by $\Phi_0^\top$ and postmultiplying by $\Phi_0$, then substituting the expression for $Q_{\phi}$ yields
   $$\Phi_0^\top\Omega\Phi_0=\Phi_0^\top \mathcal{A}_F^\top \Sigma\Phi_0 +\Phi_0^\top\Sigma \mathcal{A}_F\Phi_0+Y_0^\top QY_0.$$
   Adding the known terms $U_0^\top\mathcal{B}_F^\top \Sigma \Phi_0+\Phi_0^\top \Sigma\mathcal{B}_F U_0$ to both sides yields
   $$\Phi_0^\top\Omega\Phi_0+U_0^\top\mathcal{B}_F^\top \Sigma \Phi_0+\Phi_0^\top \Sigma\mathcal{B}_F U_0=(\mathcal{A}_F\Phi_0+\mathcal{B}_FU_0)^\top \Sigma\Phi_0 +\Phi_0^\top\Sigma (\mathcal{A}_F\Phi_0+\mathcal{B}_FU_0)+Y_0^\top QY_0.$$
   Substituting $\Phi_1=\mathcal{A}_F\Phi_0+\mathcal{B}_FU_0$ into the above equation yields
   \begin{equation}
   	\Phi_0^\top\Omega\Phi_0=(\Phi_1-\mathcal{B}_F U_0)^\top \Sigma\Phi_0 +\Phi_0^\top\Sigma (\Phi_1-\mathcal{B}_F U_0)+Y_0^\top QY_0,
   \end{equation}
   in which $\Omega$ is the sole unknown when $\Sigma$ is known. By the full row rank property of $\Phi_0$, this equation is directly solvable. Thus, $\Omega-\Sigma \mathcal{B}_F R^{-1}\mathcal{B}_F\Sigma$ serves as the increment in VI. Upon convergence of $\Sigma$, the approximate optimal feedback gain is computed via (\ref{full-K}). The complete off-policy algorithm is presented in Algorithm \ref{al2}. Consistent with classical continuous-time VI, to prevent excessive growth of the cost parameter matrix $\Sigma$, we define a sequence of constraint sets $\{\mathcal{G}_j\}_{j=0}^{\infty}$ for the increment matrix and a step-size sequence $\{\delta_i\}_{i=0}^{\infty}$ satisfying the following condition.
   \begin{equation}\label{Set-Step}
   	\mathcal{G}_{j}\subset\mathcal{G}_{j+1},\quad \lim\limits_{j\rightarrow\infty}\mathcal{G}_{j}=\mathbb{S}_{+}^{(m+2)n},\quad \sum_{i=0}^{\infty}\delta_i=\infty,\quad \sum_{i=0}^{\infty}\delta_i^2<\infty.
   \end{equation}
   
   \begin{algorithm}
   	\caption{Model-Free Value Iteration (VI)}
   	\begin{algorithmic}
   		\State \textbf{Data Collection:} Apply a $2n$-th order PE input to system (\ref{system}), and collect the corresponding input $u_t$ and output $y_t$ trajectories to form $U_0$ and $Y_0$. Concatenate $u_t$, $y_t$, and a vector of zeros into an extended vector, and use it as the input for system (\ref{eta-sys}) to collect the corresponding state $\eta_t$ and state derivative $\dot{\eta}_t$, forming $E_0$ and $E_1$. Calculate the projection matrix $F_1$, and construct the full row rank data matrices $\Phi_0=F_1E_0$, $\Phi_1=F_1E_1$, and $\mathcal{B}_F=F_1\mathcal{B}_s$.
   		\State \textbf{Initialization:} Select a precision constant $\epsilon$. Set $i=0$, $j=0$, and arbitrarily select an initial cost parameter matrix $\Sigma^0\in\mathbb{S}^{(m+2)n}_{+}$. Choose the sequence of sets $\{\mathcal{G}_j\}_{j=0}^{\infty}$ and step-size sequence $\{\delta_i\}_{i=0}^{\infty}$ that satisfy equation (\ref{Set-Step}).
   		\For{$i=0,1,2,\cdots,$}
   		\State (i) Calculate $\Omega^i$ by
   		\begin{equation}\label{VIi}
   			\Phi_0^\top\Omega^i\Phi_0=(\Phi_1-\mathcal{B}_F U_0)^\top \Sigma^i\Phi_0 +\Phi_0^\top\Sigma^i (\Phi_1-\mathcal{B}_F U_0)+Y_0^\top QY_0.
   		\end{equation}
   		\State (ii) Update the cost parameter matrix $\Sigma^{i+1}$ by $\Sigma^{i+1}=\Sigma^{i}+\delta_i(\Omega^i-\Sigma^i \mathcal{B}_F R^{-1}\mathcal{B}_F^\top\Sigma^i)$.
   		\If{$\Sigma^{i+1}\notin\mathcal{F}_{j}$}
   		\State $\Sigma^{i+1}=\Sigma^0$, $j=j+1$.
   		\EndIf
   		\If{$\Vert \Sigma^{i+1}-\Sigma^{i}\Vert/\delta_i\leq\epsilon$}
   		\State \Return $K_{VI}^*=R^{-1}\mathcal{B}_F^\top\Sigma^{i+1}$.
   		\EndIf
   		\EndFor
   	\end{algorithmic}
   	\label{al2}
   \end{algorithm}
   
   \begin{proposition}
   	Under Assumption \ref{as1}, select the eigenvalues of $A-LC$ to lie in the open left-half plane. If the initial cost parameter matrix $\Sigma^0$ is chosen to be positive semi-definite, the sequence $\{\Sigma^i\}$ generated by Algorithm \ref{al2} converges to the optimal solution $\Sigma^*$.
   \end{proposition}
   \begin{proof}
   	By the full row rank property of $\Phi_0$ and the preceding analysis, equation (\ref{VIi}) is equivalent to computing the model-based term $\Omega^{i}=\mathcal{A}_F^\top\Sigma^{i}+\Sigma^{i}\mathcal{A}_F+Q_{\phi}$.
   	Thus, Algorithm \ref{al2} corresponds to model-based VI \cite{Jiang-iter} for the equivalent new LQR problem. Since \cite{Huang-newLQR} establishes that model-based VI converges under only stabilizability and detectability, Lemmas \ref{le4} and \ref{le5} imply that the cost parameter matrices generated by Algorithm \ref{al2} converge to the respective optimal solution.
   \end{proof}
   
   	\begin{remark}
   	Noting that in classic literature on continuous-time direct data-driven control \cite{Jiang-2012,Jiang-2016,SAAA-CT-output}, integral forms of data are often used to avoid numerical differentiation of signals, we may similarly replace the data matrices $\Phi_0$, $U_0$, $Y_0$, and $\Phi_1$ with their corresponding integral forms $\tilde{\Phi}_0:=\int_{0}^{\Delta t} [\phi_{\tau},\cdots,\phi_{\tau+(T-1)\Delta t}]d\tau$, $\tilde{U}_0:=\int_{0}^{\Delta t} [u_{\tau},\cdots,u_{\tau+(T-1)\Delta t}]d\tau$, $\tilde{Y}_0:=\int_{0}^{\Delta t} [y_{\tau},\cdots,y_{\tau+(T-1)\Delta t}]d\tau$, and $\tilde{\Phi}_1:=[\phi_{\Delta t}-\phi_0,\cdots,\phi_{T\Delta t}-\phi_{(T-1)\Delta t}]$, as they still satisfy
   	$$\tilde{\Phi}_1=F_1(\mathcal{A}_s+\mathcal{B}_s^yCS)F_1^\top \tilde{\Phi}_0+F_1\mathcal{B}_s^u \tilde{U}_0,\quad \tilde{Y}_0=CSF_1^\top \tilde{\Phi}_0.$$
   	Nevertheless, this paper maintains the use of the derivative matrix $\Phi_1$. This is because $\dot{\eta}_t$ can be obtained directly from (\ref{eta-sys}) without model knowledge, and the projection matrix $F_1$ is constructed directly from the data matrix $E_0$. Consequently, $\dot{\phi}_t=F_1\dot{\eta}_{t}$ can be computed directly and accurately, eliminating the need for integral data matrices, which would introduce additional numerical errors and computational complexity.
   \end{remark}

   Now we elucidate the efficiency of Algorithms \ref{al1} and \ref{al2}. First, the proposed algorithms only require $\Phi_0$ to be of full row rank, which is guaranteed by a PE input and Lemma \ref{le2}. This renders the algorithms feasible for both single-output and multi-output scenarios, a distinct advantage over existing LS-based PI and VI algorithms. Works such as \cite{SAAA-CT-output,Deng-CT-output,Huang-newLQR} reformulate the iterative equation as an LS problem, which strictly requires the regression matrix, formed by the Kronecker product (or further integrals) of data, to be full row rank. We can use assertion (b) of Lemma \ref{le1} to demonstrate that this rank condition is difficult to achieve in multi-output scenarios. Specifically, ensuring a full row rank regression matrix generally requires the matrix $\Gamma_{\eta}:=[\mathrm{vec}(\eta_t\eta_t^\top),\cdots,\mathrm{vec}(\eta_{t+(T-1)\Delta t}\eta_{t+(T-1)\Delta t}^\top)]\in\mathbb{R}^{(m+p+1)^2n^2\times T}$ to be of full row rank. By assertion (b) of Lemma \ref{le1}, the span of $\eta_{t}$ is $\mathcal{V}=\mathrm{im}(F_1^\top)$. We thus treat the columns of $F_1^\top$ as a basis for $\mathcal{V}$ and set $\eta_{t}=F_1^\top b_t$, where $b_t\in\mathbb{R}^{(m+2)n}$ is a coefficient vector. Then,
   $$\mathrm{vec}(\eta_t\eta_{t}^\top)=\mathrm{vec}(F_1^\top b_tb_t^\top F_1)=(F_1^\top\otimes F_1^\top)\mathrm{vec}(b_tb_t^\top).$$
   Hence,
   $$\mathrm{rank}(\Gamma_{\eta})\leq\mathrm{rank}([\mathrm{vec}(b_tb_t^\top),\cdots,\mathrm{vec}(b_{t+(T-1)\Delta t}b_{t+(T-1)\Delta t}^\top)])\leq (m+2)^2n^2\leq (m+p+1)^2n^2,$$
   where the last inequality holds strictly when system (\ref{system}) is  multi-output. Consequently, $\Gamma_{\eta}$ cannot be of full row rank, implying the regression matrix does not have full row rank either, which may cause LS-based algorithms to fail. Second, the proposed algorithms involve fewer unknowns in their iterative equations. Since this paper reformulates the original LQR problem as an equivalent new LQR problem with accessible states, the unknown parameters in the iterative equations no longer include terms related to the control gain. Additionally, when $p>1$, projecting the input-output filtering vector $\eta_{t}$ using $F_1$ further eliminates redundant information, thereby reducing the number of unknowns in the iterative equations by $(p-1)n((2m+p+3)n+1)/2$. Third, the proposed algorithms require less data. LS-based algorithms necessitate the amount of data to exceed the number of unknowns; for instance, $T\geq (m+p+1)n((m+p+1)n+1)/2+(m+p+1)nm$ in \cite{Deng-CT-output} and $T\geq (m+p)n((m+p)n+1)/2$ in \cite{Huang-newLQR}. In contrast, the proposed algorithms only require the number of columns of $\Phi_0$ to be no less than the number of rows, i.e., $T\geq (m+2)n$.

   \section{Further Discussion}\label{sec5}
   Building upon the content of Sections \ref{sec2_4} and \ref{sec3_1}, this section briefly discusses the trajectory generation for continuous-time linear systems, the PE conditions required in this paper, the determination of the system order $n$, and potential policy optimization improvements enabled by the full row rank matrix $\Phi_0$.

   \textbf{Trajectory generation for continuous-time systems.} Existing research typically follows a fixed paradigm: it extends the weighting vector $\alpha$ in (\ref{alp}) to a time-varying vector $\alpha_t$, constrains $\alpha_t$ via ordinary differential equations with coefficients derived from input-state or input-output data matrices, and generates true trajectories by weighting sampled input-output data matrices with $\alpha_t$. However, current methods lack universality: the work in \cite{CT-PCPE} is limited to systems with measurable states; the approach in \cite{In-out} utilizes data matrices composed of inputs, outputs, and their derivatives, but as shown in equation (\ref{rank-uy}), such a matrix is not of full row rank when the output dimension $p>1$, whereby the corresponding ordinary differential equations cannot be directly solved, making the method invalid. To address this, based on assertion (b) of Lemma \ref{le2}, Proposition \ref{le8} provides constructive conditions for generating valid true trajectories of system (\ref{system}) in multi-output scenarios.
   This resolves the open problem in \cite{In-out}, requiring no state measurability and avoiding high-order numerical differentiation of continuous-time signals.

   	\begin{proposition}\label{le8}
   	Assume the input signal $\{u_t\}_{t\in\mathbb{R}_+}$ of system (\ref{system}) satisfies the $(2n+1)$-th order PE condition, and let $\{y_t\}_{t\in\mathbb{R}_{+}}$ be the corresponding output signal. Construct the full row rank substitute state data matrix $\Phi_0$ at time $t$ according to Lemma \ref{le2}. Then, the following assertions hold.
   	\begin{itemize}
   		\item[\textup{(a)}] If there exists a vector $\alpha_t\in\mathbb{R}^{T}$ such that $
   		\Phi_0\alpha_t^{(1)}=\mathbf{0}$ holds, then the continuously differentiable vector $[(\bar{u}_t)^\top,(\bar{y}_t)^\top]^\top$ satisfying equation (\ref{CT-H}) is a true trajectory of system (\ref{system}).
   		\begin{equation}\label{CT-H}
   			\begin{bmatrix}
   				H_{1}(u_{(t,\Delta t,T-1)})\\
   				H_{1}(y_{(t,\Delta t,T-1)})
   			\end{bmatrix}\alpha_t=\begin{bmatrix}
   				\bar{u}_t\\	\bar{y}_t
   			\end{bmatrix}.
   		\end{equation}
   		\item[\textup{(b)}] For a continuously differentiable trajectory $[(\bar{u}_t)^\top,(\bar{y}_t)^\top]^\top$ of system (\ref{system}), if there exists a vector $\alpha_t\in\mathbb{R}^{T}$ such that the following equation holds,
   		$$\begin{bmatrix}
   			H_{1}(u_{(t,\Delta t,T-1)})\\\Phi_0
   		\end{bmatrix}\alpha_t^{(1)}=-\begin{bmatrix}
   			H_{1}(u^{(1)}_{(t,\Delta t,T-1)})\\\mathbf{0}
   		\end{bmatrix}\alpha_t+\begin{bmatrix}
   			\bar{u}_t^{(1)}\\\mathbf{0}
   		\end{bmatrix},$$
   		then $\bar{y}_t$ can be derived from equation (\ref{CT-H}).
   	\end{itemize}		
   \end{proposition}
   
   \begin{proof} 	
   	For assertion (a), defining $\bar{x}_t:=H_{1}(x_{(t,\Delta t,T-1)})\alpha_t$, the linearity of the delayed Hankel matrix yields
   	\begin{align*}
   		\dot{\bar{x}}_t&=\frac{d}{dt}(H_{1}(x_{(t,\Delta t,T-1)})\alpha_t)=H_{1}(\dot{x}_{(t,\Delta t,T-1)})\alpha_t+H_{1}(x_{(t,\Delta t,T-1)})\alpha^{(1)}_t\\
   		&=(AH_{1}(x_{(t,\Delta t,T-1)})+BH_{1}(u_{(t,\Delta t,T-1)}))\alpha_t+H_{1}(x_{(t,\Delta t,T-1)})\alpha^{(1)}_t\\
   		&=AH_{1}(x_{(t,\Delta t,T-1)})\alpha_t+BH_{1}(u_{(t,\Delta t,T-1)})\alpha_t+SF_1^\top\Phi_0\alpha^{(1)}_t\\
   		&=A\bar{x}_t+B\bar{u}_t.
   	\end{align*}
    Denote the substitute state data matrix at the initial time $t_0$ as $\Phi_0^{ini}$, and the initial weighting vector as $\alpha_{t_0}=(\Phi_0^{ini})^{\dagger}F_1\bar{\eta}_{t_0}$, where $\bar{\eta}_{t_0}$ is the filtering vector obtained by driving (\ref{eta-sys}) with $[(\bar{u}_{t_0})^\top,(\bar{y}_{t_0})^\top,\mathbf{0}]^\top$.
   	Consequently, under the initial state $\bar{x}_{t_0}=H_{1}(x_{(t_0,\Delta t,T-1)})\alpha_{t_0}=SF_1^\top\Phi_0^{ini}\alpha_{t_0}$ and input $\bar{u}_t$, $\bar{x}_t$ is the state trajectory of system (\ref{system}). Furthermore, since
   	\begin{align*}
   		\bar{y}_t=H_{1}(y_{(t,\Delta t,T-1)})\alpha_t=CH_{1}(x_{(t,\Delta t,T-1)})\alpha_t=C\bar{x}_t,
   	\end{align*}
   	$\bar{y}_t$ is the output trajectory of system (\ref{system}).
   	
   	For assertion (b), it is straightforward to show that $\bar{u}_t=H_{1}(u_{(t,\Delta t,T-1)})\alpha_t$. Furthermore, since $\Phi_0\alpha_t^{(1)}=\mathbf{0}$, we have
   	\begin{align*}
   		\dot{\bar{x}}_t&=A\bar{x}_t+B\bar{u}_t=A\bar{x}_t+B\bar{u}_t+SF_1^\top\Phi_0\alpha_t^{(1)}\\
   		&=A\bar{x}_t+H_{1}(u_{(t,\Delta t,T-1)})\alpha_t+H_{1}(x_{(t,\Delta t,T-1)})\alpha_t^{(1)}.
   	\end{align*}	
   	From the initial conditions, it follows that $\bar{x}_t=H_{1}(x_{(t,\Delta t,T-1)})\alpha_t$, which further implies that
   	\begin{align*}
   		\bar{y}_t=C\bar{x}_t=CH_{1}(x_{(t,\Delta t,T-1)})\alpha_t=H_{1}(y_{(t,\Delta t,T-1)})\alpha_t.
   	\end{align*}
   	This completes the proof.
   \end{proof}

  \textbf{Piecewise constant PE.}
  Section \ref{sec2_4} defines the continuous-time PE condition using high-order derivatives of signals to derive (\ref{Hux}). In contrast, the work in \cite{CT-PCPE} directly defines a continuous-time input as $N$-th order PE if it satisfies
  $$\mathrm{rank}\left(\begin{bmatrix}
  	H_{N}(u_{(t,\Delta,N+T-1)})\\H_{1}(u_{(t,\Delta,T-1)})
  \end{bmatrix}\right)=mN+n,$$
  and shows that piecewise constant PE inputs satisfy this definition. Piecewise constant PE inputs do not meet the definition adopted herein, but they relax continuity and smoothness requirements while still ensuring a similar full row rank property of the input-state matrix, which relies on delayed signals instead of derivatives. It can be verified via experiments that under piecewise constant PE, the row rank of matrix $E_0$ remains $(m+p+1)n$. Consequently, the full row rank data matrix $\Phi_0$ can still be constructed in the same manner and utilized for model-free optimal controller design. This suggests that the required excitation does not necessitate global high-order continuous differentiability, but only piecewise continuous differentiability. The theoretical analysis of such relaxed continuous-time PE conditions is a direction for future research.

\textbf{Dimension $n$ of the unmeasurable state $x_t$.} Algorithms \ref{al1} and \ref{al2} do not require prior knowledge of the state dimension $n$, as the full QR decomposition naturally reveals the full row rank property of the data matrix $\Phi_0$, confirming its row rank to be $(m+2)n$. Furthermore, the order $n$ of the linear system with unknown dynamics and unmeasurable states can be obtained by dividing the row rank of $\Phi_0$ by the known quantity $m+2$. Compared to calculating $n$ using equation (\ref{rank-uy}), this method avoids the need for numerical differentiation of the input-output signals.

\textbf{Data-enabled policy gradient.}
For the model-free discrete-time LQR problem with measurable states, the corollary (\ref{IK}) of Lemma \ref{willems} has been successfully used to design the data-enabled policy gradient algorithm \cite{Deepo}. Motivated by the structural analogy between the differential operator of continuous-time systems and the shift operator of discrete-time systems, together with the facts that $\phi_{t}$ serves as a substitute state for $x_t$ with an injective linear relationship and that the feedback control gain $K$ is intrinsically related to the full row rank input-substitute state data matrix (as shown in Theorem \ref{thm2} and analogous to (\ref{IK})), the results developed in this paper pave the way for extending model-free policy gradient algorithms to continuous-time settings.

\section{Numerical Experiments}\label{sec6}
This section validates the computational efficiency of Algorithms \ref{al1} and \ref{al2} for single-output problems, as well as their applicability to multi-output problems, through numerical experiments.

\subsection{Single Output Problem}\label{sec6_1}
This subsection initializes the MATLAB random number generator using the $rng(3)$ function. Randomly generate $50$ continuous-time single-input single-output linear time-invariant systems with dimensions $n=2$, $m=1$, $p=1$ that satisfy Assumption \ref{as1}. The dynamics matrix $A$ is a symmetric matrix with all eigenvalues within the interval $(-2,0)$, while the input matrix $B$ and output matrix $C$ are strictly positive matrices. Set $Q=1$ and $R=1$. 
For each randomly generated LQR problem, the eigenvalues of $A-LC$ are chosen as $-5$ and $-6$, and $A_{\epsilon}$ is selected as a controllable canonical form matrix with the same eigenvalues. Each element of
$x_0$ and $\eta_0^{\epsilon}$ independently follows a uniform distribution over the interval $[-1,1]$. 
The required data are collected using the PE input $u_t= -0.2+0.2(\sin(74\pi t)+1.5\sin(38\pi t)+2\sin(26\pi t)+2.5\sin(10\pi t))$ proposed in \cite{Deng-CT-output}. For the PI and VI algorithms that do not include replication of observation errors proposed in \cite{Huang-newLQR} (LS-Based Transformed PI and LS-Based Transformed VI), data collection commences at $t=12s$ with a sampling step of $\Delta t=0.04s$. For the VI algorithm proposed in \cite{Deng-CT-output} (LS-Based VI)  and the algorithms proposed in this paper, sampling begins at the initial time $t=0$ with a sampling step of $\Delta t=0.2s$.
For all algorithms, the precision threshold is set to $\epsilon=0.01$. The initial stabilizing controller $K^0$ for all PI algorithms is selected as the zero matrix of the corresponding dimension, with a maximum of $100$ iterations. The initial cost parameter matrix for all VI algorithms is selected as the identity matrix of the corresponding dimension, with a maximum of $3000$ iterations; the step size is defined as $\{\delta_i\}_{i=0}^{\infty}=10/(i+1000)$, and the sequence of constraint sets for the incremental matrix is $\{\mathcal{G}_j\}_{j=0}^{\infty}=\{\Sigma\in\mathbb{S}_{+}\big\vert \Vert \Sigma\Vert<10^5(j+1)\}$. 
Since a fixed random number generator and PE input are used in the comparative experiments, the regression matrices required by the LS-based algorithms may not necessarily be of full row rank. Therefore, Tables \ref{ta1} and \ref{ta2} only compare the average number of iterations, average computation time, and the average relative error of the resulting feedback control gain, defined as $\frac{\Vert K^{*}_{PI/VI}-K^*\Vert}{\Vert K^*\Vert}$, among the five algorithms for cases where the regression matrix has full row rank. Figure \ref{f1} illustrates the distribution of the non-zero minimum singular values of the obtained data matrices for the $50$ randomly generated LQR problems.

\begin{table}[h]
	\centering
	\begin{tabular}{|c|c|c|}
		\hline
		&LS-Based Transformed PI& Algorithm \ref{al1}\\\hline
		Number of full row rank data matrices / 50&$36$&$50$\\\hline
		Average number of iterations&$4.3889$&$3.700$\\\hline
		Average computation time (s)&$5.6941\times 10^{-4}$&$5.0123\times 10^{-4}$\\\hline
		Average relative error 	&$0.0393$&$6.1397\times 10^{-10}$\\\hline
	\end{tabular}
	\caption{Performance Comparison of PI Algorithms.}
	\label{ta1}
\end{table}

\begin{table}[h]
	\centering
	\begin{tabular}{|c|c|c|c|}
		\hline
		&LS-Based VI&LS-Based Transformed VI& Algorithm \ref{al2}\\\hline
		Number of full row rank data matrices / 50&$42$&$43$&$50$\\\hline
		Average number of iterations&$1678.0476$ &$277.6279$&$137.4000$\\\hline
		Average computation time (s)&$0.0973$ &$0.0189$ &$0.0025$\\\hline
		Average relative error 	&$0.0024$  &$0.2002$  &$1.0031\times 10^{-6}$\\\hline
	\end{tabular}
	\caption{Performance Comparison of VI Algorithms.}
	\label{ta2}
\end{table}

\begin{figure}[h]
	\centerline{\includegraphics[width=0.7\textwidth]{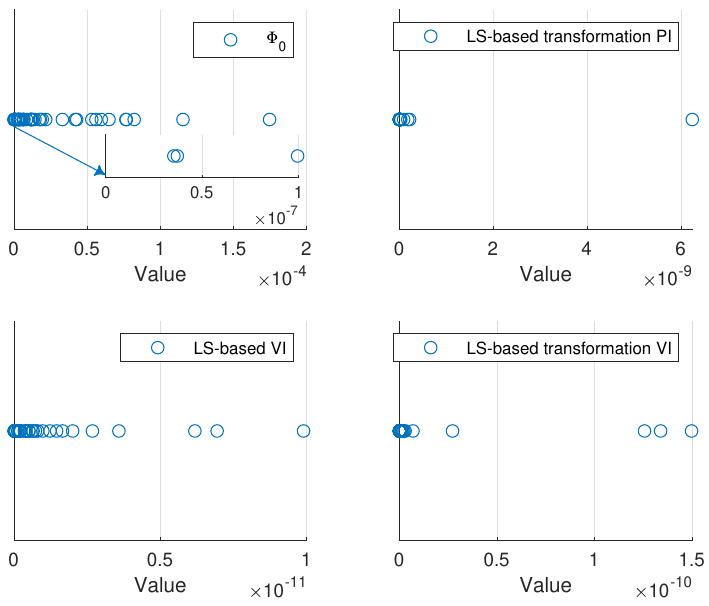}}
	\caption{Comparison of Minimum Singular Values of Data Matrices.}
	\label{f1}
\end{figure}

Based on the experimental results, during the data collection phase, Algorithms \ref{al1} and \ref{al2} allow sampling at any time instant, eliminating the need to wait for the system to reach steady state or to perform Kronecker product operations on the data, thereby reducing the time required for data collection. 
In terms of numerical stability, the data matrix $\Phi_0$
utilized by Algorithms \ref{al1} and \ref{al2} is guaranteed to be of full row rank, with the smallest minimum singular value across the $50$ tests found to be $3.5309\times 10^{8}$. In contrast, the three LS-based algorithms encounter cases where the data matrices are not of full row rank, which requires data regeneration by adjusting random number seeds or PE inputs. Furthermore, the minimum singular values of the full row rank regression matrices for these LS-based methods are all below the order of $10^{-8}$, which introduces computational instability. 
In terms of computational efficiency, Algorithm \ref{al1} demonstrates slight superiority over the LS-Based Transformed PI algorithm. Among the VI algorithms, Algorithm \ref{al2} exhibits significant advantages, requiring fewer iterations and less time to achieve the same convergence precision.
For the optimality of the resulting feedback control gains, Algorithms \ref{al1} and \ref{al2} yield significantly smaller average relative errors. Conversely, the algorithms proposed in \cite{Huang-newLQR} produce larger final relative errors in the tests on randomly generated LQR problems. Specifically, the LS-Based Transformed PI algorithm yields $5$ results with relative errors greater than $0.1$, and the LS-Based Transformed VI algorithm yields $13$ such results, while the remaining results have very small relative errors. This behavior occurs because these two algorithms do not adequately account for the influence of observation errors, which can lead to large errors if data collection is improper. Therefore, the algorithms proposed in this paper exhibit superior performance.

\subsection{Multiple Output Problem}\label{sec6_2}
Consider a continuous-time LQR problem with dimensions $n=4$, $m=2$, and $p=2$, where
\begin{equation}\label{MO}
	\begin{split}
		&	A=\begin{bmatrix}
			-1.0169&-1.4786&1.7280&0.2547\\
			1.5194&-0.5787&0.3642&0.1249\\
			-1.6774&0.4439&-0.6473&-0.3487\\
			-0.3387&-0.2713&0.1172&-0.7287
		\end{bmatrix},\quad B=\begin{bmatrix}
			-0.2938& 0\\
			-0.8479& 0.3075\\
			-1.1201&-1.2571\\
			0 &1
		\end{bmatrix}\\
		& C=\begin{bmatrix}
			1&1&0&0\\0&1&0&0
		\end{bmatrix},\quad Q=I_{2},\quad R=2I_{2}.
	\end{split}
\end{equation}
The eigenvalues of $A-LC$ are selected as $-2$, $-3$, $-4$, and $-5$, and $A_{\epsilon}$ is chosen as a controllable canonical form matrix with the same eigenvalues. Each element of $x_0$ and $\eta_0^{\epsilon}$ independently follows a uniform distribution over the interval $[-10,10]$. The required data are collected using the PE input $u_t=[u_t^1,u_t^2]^\top$, where $u_t^1=0.5\sin(2\pi t)+0.3\cos(4\pi t)+0.2\sin(6\pi t+\pi/4)+0.1\cos(10\pi t)$ and $u_t^2=0.4\cos(3\pi t)+0.25\sin(5\pi t) +0.15\cos(7\pi t-\pi/3) +0.05\sin(9\pi t)$, with a sampling step of $\Delta t=0.2s$. 
Without fixing the random number generator, $100$ data collection trials are conducted to satisfy the feasibility condition of the LS-based algorithms, i.e., the regression matrix being of full row rank. The results show that the average row rank of the regression matrices for the LS-Based VI algorithm proposed in \cite{Deng-CT-output} is
$115.84$ (with $250$ rows); for the LS-Based Transformed PI and LS-Based Transformed VI algorithms proposed in \cite{Huang-newLQR}, the average row ranks of their regression matrices are $20.03$ and $21.56$ (with $136$ rows), respectively. None of these three algorithms achieved a full row rank regression matrix in any of the $100$ trials. This indicates that the existing LS-based algorithms lack practical feasibility for multi-output problems.

Now we consider the algorithms proposed in this paper. The MATLAB random number generator is initialized using the $rng(3)$ function, and the precision threshold for Algorithms \ref{al1} and \ref{al2} is set to $\epsilon=0.01$.
Given that the dynamics matrix $A$ of the LQR problem (\ref{MO}) is Hurwitz stable, the initial stabilizing controller for Algorithm \ref{al1} can be selected as $K^0=\mathbf{0}_{m\times (m+2)n}$. Algorithm \ref{al1} terminates after $7$ iterations, with an iterative residual of $\Vert K^{7}-K^{6}\Vert=3.8273\times 10^{-6}$ at termination. The curve of the relative error $\frac{\Vert K^{i}-K^*\Vert}{\Vert K^*\Vert}$ against the number of iterations is shown in Figure \ref{f2}. The final relative error achieved is $\frac{\Vert K_{PI}^*-K^*\Vert}{\Vert K^*\Vert}=3.8273\times 10^{-9}$, and the calculated optimal cost is $\phi_0^\top \Sigma_{PI}^{*}\phi_0=11.4715$, with a difference of $8.6535\times 10^{-11}$ from the true optimal cost $x_0^\top P_x^* x_0=11.4715$. Therefore, the solution $K_{PI}^*$ obtained by Algorithm \ref{al1} is a nearly optimal solution to the LQR problem (\ref{MO}).

\begin{figure}[h]
	\centerline{\includegraphics[width=0.7\textwidth]{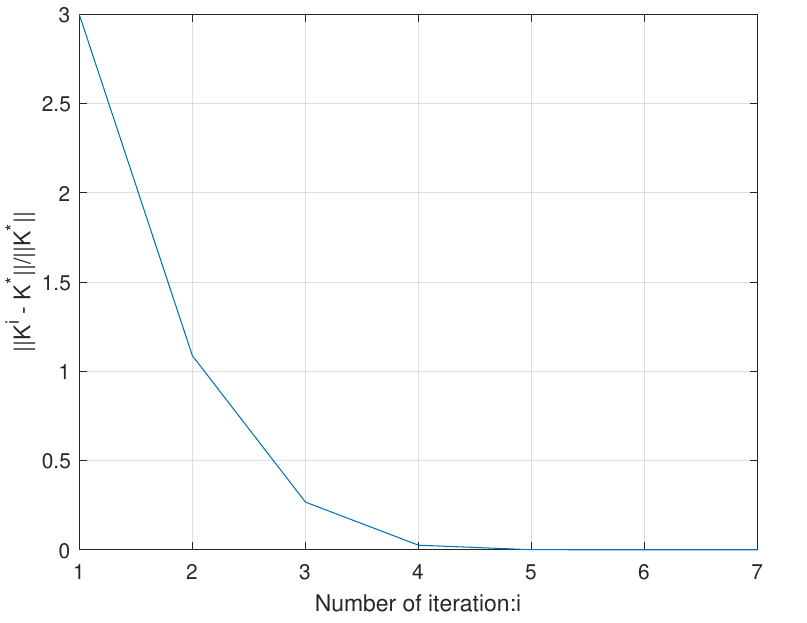}}
	\caption{Iteration curve of the relative error for Algorithm \ref{al1}.}
	\label{f2}
\end{figure}

If the element $-0.5787$ at the $(2,2)$ position of the dynamics matrix $A$ in the LQR problem (\ref{MO}) is modified to $0.5787$, the dynamics matrix $A$ is no longer Hurwitz stable, but Assumption \ref{as1} is still satisfied.
Algorithm \ref{al2} is applied to the modified LQR problem with the altered $A$. The initial cost parameter matrix is selected as $\Sigma^0=0.01I_{(m+2)n}$, the step size is set to $\{\delta_i\}_{i=0}^{\infty}=100/(i+1000)$, and the sequence of constraint sets for the incremental matrix is $\{\mathcal{G}_j\}_{j=0}^{\infty}=\{\Sigma\in\mathbb{S}_{+}^{(m+2)n}\big\vert \Vert \Sigma\Vert<10^5(j+1)\}$. Algorithm \ref{al2} terminates after $155$ iterations, with an iterative residual of $\Vert P^{155}-P^{154}\Vert/\delta_{154}=8.9369\times 10^{-3}$ at termination. The curve of the relative error $\frac{\Vert K^{i}-K^*\Vert}{\Vert K^*\Vert}$ against the number of iterations is shown in Figure \ref{f3}. The final relative error achieved is $\frac{\Vert K_{VI}^*-K^*\Vert}{\Vert K^*\Vert}=5.1723\times 10^{-8}$, and the calculated optimal cost is  $\phi_0^\top \Sigma_{VI}^{*}\phi_0=17.7086$, with a difference of $1.1727\times 10^{-5}$ from the true optimal cost $x_0^\top P_x^* x_0=17.7086$. Therefore, the solution $K_{VI}^*$ obtained by Algorithm \ref{al2} is a nearly optimal solution to the original problem.

\begin{figure}[h]
	\centerline{\includegraphics[width=0.7\textwidth]{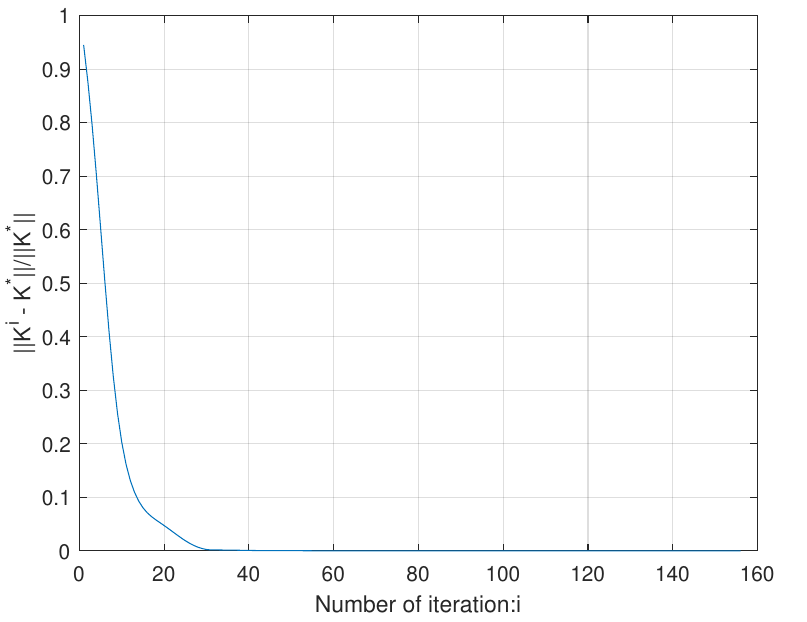}}
	\caption{Iteration curve of the relative error for Algorithm \ref{al2}.}
	\label{f3}
\end{figure}

The controllability assumption in Assumption \ref{as1} is imposed primarily to facilitate the derivation of the theoretical results in this paper. If system (\ref{system}) is stabilizable but not controllable, Algorithms \ref{al1} and \ref{al2} remain feasible provided that suitable user-defined parameters are chosen such that both the state parameterization matrix $S$ and the data matrix $\Phi_0$ have full row rank. If we set $A=\mathrm{diag}(-1,-2,-3,-4)$,
$$B=\begin{bmatrix}
	1&0&0&0\\0&1&1&0
\end{bmatrix}^\top,\quad C=\begin{bmatrix}
	1&0& 1&0\\ 0& 1& 0& 1
\end{bmatrix},$$
then the linear system satisfies observability and stabilizability but is uncontrollable. If the cost weighting matrices and user-defined parameters remain unchanged, the matrices $S$ and $\Phi_0$ obtained thereby still have full row rank.
Applying Algorithm \ref{al1} to this uncontrollable LQR problem, the algorithm terminates after $4$ iterations with an iterative residual of $\Vert K^{4}-K^{3}\Vert=8.347\times 10^{-6}$. The final relative error achieved is $\frac{\Vert K_{PI}^*-K^*\Vert}{\Vert K^*\Vert}=1.2689\times 10^{-8}$, and the calculated optimal cost is $\phi_0^\top \Sigma_{PI}^{*}\phi_0=5.8385$, differing from the true optimal cost $x_0^\top P_x^* x_0=5.8385$ by $5.4712\times 10^{-13}$. Applying Algorithm \ref{al2} to this uncontrollable LQR problem, the algorithm terminates after $53$ iterations with an iterative residual of $\Vert K^{53}-K^{52}\Vert=7.7395\times 10^{-3}$. The final relative error achieved is $\frac{\Vert K_{VI}^*-K^*\Vert}{\Vert K^*\Vert}=4.3403\times 10^{-7}$, and the calculated optimal cost is $\phi_0^\top \Sigma_{VI}^{*}\phi_0=5.8385$, differing from the true optimal cost $x_0^\top P_x^* x_0=5.8385$ by $3.6101\times 10^{-7}$. These results demonstrate that Algorithms \ref{al1} and \ref{al2} can still converge to a nearly optimal solution for uncontrollable LQR problems.

	\section{Conclusion}\label{sec7}		
In the continuous-time framework, this work proposes an effective substitute state construction method that guarantees the data matrix to be of full row rank. Based on this, the continuous-time LQR problem with unmeasurable states and unknown model parameters is equivalently transformed into a new LQR problem with the proposed substitute state as the state variable and a known input matrix. Efficient model-free PI and VI algorithms are then designed to solve for the optimal feedback controller, which eliminate hard-to-satisfy feasibility conditions, avoid reliance on LS solutions, and provide solid theoretical guarantees.

The data-driven theory proposed in this paper for unmeasurable state settings exhibits strong generalizability. Future work will focus on exploring its integration with other reinforcement learning algorithms and its application to other optimal control problems. On the other hand, while the data used herein are exact, real-world data is typically corrupted by noise, making algorithm design and theoretical analysis under noisy data a subject for future research. Furthermore, the substitute states constructed in relevant literature and this paper remain redundant representations of the true states, which increases the computational burden in high-dimensional systems. Thus, developing model-free solutions for high-dimensional optimal control problems with unmeasurable states is also a key focus for our future research.

	\bibliographystyle{unsrt}
	\bibliography{CT-output-ref}           

\end{document}